\documentclass[a4paper]{article}

\DeclareFontShape{OT1}{cmr}{m}{scit}{<->ssub*cmr/m/sc}{}
\usepackage{slantsc}

\usepackage[margin=1in]{geometry} %

\usepackage{amsmath}
\usepackage{amsthm}
\usepackage{amssymb}
\usepackage{graphicx}
\usepackage{xcolor}
\usepackage{mathrsfs} %
\usepackage[utf8]{inputenc}
\usepackage{adjustbox}
\usepackage{enumerate}
\usepackage{verbatim}

\usepackage{tikz}
\usetikzlibrary{shapes,calc,math,backgrounds,matrix}

\usepackage{hyperref}
\hypersetup{
	colorlinks,
	allcolors=blue,
	unicode,
	bookmarksopen,
	bookmarksdepth=2,
	pdftitle={The Complexity of Distance-$r$ Dominating Set Reconfiguration},
	pdfauthor={Niranka Banerjee, Duc A. Hoang},
	pdfkeywords={Distance-$r$ dominating set, Reconfiguration problem, Computational complexity, PSPACE-completeness, Polynomial-time algorithm}
}

\usepackage[capitalise,noabbrev]{cleveref} %

\usepackage[%
backend=biber, 
bibstyle=numeric, 
citestyle=numeric-comp,
maxcitenames=3,
maxbibnames=20,
sorting=ydnt
]{biblatex}
\theoremstyle{plain}
\newtheorem{theorem}{Theorem}[section]
\newtheorem{corollary}[theorem]{Corollary}
\newtheorem{lemma}[theorem]{Lemma}
\newtheorem{claim}{Claim}[theorem]

\newtheorem{proposition}[theorem]{Proposition}

\theoremstyle{definition}

\usepackage[linesnumbered,ruled,vlined,commentsnumbered]{algorithm2e} %

\usepackage[colorinlistoftodos]{todonotes}
\definecolor{done}{rgb}{0.55, 0.71, 0.0}
\definecolor{highpriority}{rgb}{0.82, 0.1, 0.26}
\definecolor{lowpriority}{rgb}{1.0, 0.75, 0.0}
\definecolor{mediumpriority}{rgb}{1.0, 0.49, 0.0}
\newcommand{\Mcomment}[3]{%
\ifsubmission%
	\else%
	{\color{#1}\bfseries\sffamily(#3)%
	}%
	\marginpar{\textcolor{#1}{\bfseries\sffamily #2}}%
	\fi%
}
\newcommand{\TaskNewPerson}[3]{%
	\expandafter\newcommand\csname #2\endcsname[1]{\Mcomment{#3}{#2}{##1}}%
	\expandafter\newcommand\csname Task#1\endcsname[2][mediumpriority]{\expandafter\TaskPerson[backgroundcolor=##1]{#1}{##2}}
	}

\newcommand*{\TaskPerson}[3][]{
\ifsubmission%
	\else%
	\todo[inline,#1]{\textbf{(#2)} #3}
	\fi%
}

\newcommand{\email}[1]{\href{mailto:#1}{\texttt{#1}}} %

\newcommand{\sfTS}{{\mathsf{TS}}} %
\newcommand{\sfTJ}{{\mathsf{TJ}}} %
\newcommand{\sfTAR}{{\mathsf{TAR}}} %
\newcommand{\sfR}{{\mathsf{R}}} %

\newcommand{\calG}{{\mathcal{G}}}
\newcommand{\calH}{{\mathcal{H}}}
\newcommand{\calS}{{\mathcal{S}}}

\newcommand{\ttPSPACE}{{\mathtt{PSPACE}}}
\newcommand{\ttNP}{{\mathtt{NP}}}
\newcommand{\ttP}{{\mathtt{P}}}

\newcommand{\dist}{\mathsf{dist}} %
\newcommand{\diam}{\mathsf{diam}} %
\newcommand{\opt}{\mathsf{OPT}} %
\newcommand{\reconf}[2][\sfR]{\overset{#2}{\longrightarrow}_{#1}} %

\newcommand{\ReviewRevise}[1]{{\color{red}#1}} %
\newcommand{\ExtraMod}[1]{{\color{blue}#1}} %
\renewcommand{\ReviewRevise}[1]{#1} %
\renewcommand{\ExtraMod}[1]{#1} %

\title{\textbf{The Complexity of Distance-$r$ Dominating Set Reconfiguration}\thanks{A preliminary version of this paper has appeared in the Proceedings of the 30th International Computing and Combinatorics Conference (COCOON 2024).}}
\author{Niranka Banerjee$^1$ \and Duc~A.~Hoang$^2$}
\date{
	$^1$Graduate School of Engineering, Mie University, Japan\\
              \email{banerjee@eng.mie-u.ac.jp}           \\%
	$^2$VNU University of Science,
	Vietnam National University,
	Hanoi, 
	Vietnam \\ \email{hoanganhduc@hus.edu.vn}
	\\[2ex]%
}

\begin{document}
\maketitle

\begin{abstract}
For a fixed integer $r \geq 1$, a \textit{distance-$r$ dominating set} of a graph $G = (V, E)$ is a subset $D \subseteq V$ such that every vertex in $V$ is within distance $r$ from some member of $D$. Given two distance-$r$ dominating sets $D_s$ and $D_t$ of $G$, the \textsc{Distance-$r$ Dominating Set Reconfiguration (D$r$DSR)} problem asks whether there exists a sequence of distance-$r$ dominating sets transforming $D_s$ into $D_t$, where each intermediate member is obtained from its predecessor by applying a given reconfiguration rule exactly once. The case when $r = 1$ has been extensively studied in the literature. We study \textsc{D$r$DSR} for $r \geq 2$ under two well-known reconfiguration rules: Token Jumping ($\mathsf{TJ}$), which involves replacing a member of the current D$r$DS with a non-member, and Token Sliding ($\mathsf{TS}$), which involves replacing a member of the current D$r$DS with an adjacent non-member.

For $r = 1$, it is known that under either $\mathsf{TS}$ or $\mathsf{TJ}$, the problem on split graphs is $\mathtt{PSPACE}$-complete. We show that for $r \geq 2$, the problem is in $\mathtt{P}$, resulting in an interesting complexity dichotomy. Along the way, we establish nontrivial bounds on the length of a shortest reconfiguration sequence on split graphs when $r = 2$, which may be of independent interest. Moreover, we provide observations that lead to polynomial-time algorithms for \textsc{D$r$DSR} for $r \geq 2$ on dually chordal graphs under $\mathsf{TJ}$ and on cographs under either $\mathsf{TS}$ or $\mathsf{TJ}$. We also design a linear-time algorithm for solving the problem under $\mathsf{TJ}$ on trees. On the negative side, we prove that \textsc{D$r$DSR} for $r \geq 1$ remains $\mathtt{PSPACE}$-complete on planar graphs of maximum degree three and bounded bandwidth, improving the degree bound of previously known results. We further demonstrate that the known $\mathtt{PSPACE}$-completeness results under $\mathsf{TS}$ and $\mathsf{TJ}$ on bipartite graphs and chordal graphs for $r = 1$ can be extended to $r \geq 2$.

\noindent\textbf{Keywords:} Distance-$r$ dominating set, Reconfiguration problem, Computational complexity, PSPACE-completeness, Polynomial-time algorithm

\noindent\textbf{MSC 2020:} 05C85
\end{abstract}

\section{Introduction}
\label{sec:intro}

\subsubsection*{Reconfiguration Problems.}
\textit{Reconfiguration problems} have independently emerged in various areas of computer science over the past few decades, including recreational mathematics (e.g., games and puzzles), computational geometry (e.g., flip graphs of triangulations), constraint satisfaction (e.g., solution space of Boolean formulas), and even quantum complexity theory (e.g., ground state connectivity), among others~\cite{Heuvel13,Nishimura18,MynhardtN19,BousquetMNS24}.	
In a \textit{reconfiguration variant} of a computational problem (e.g., \textsc{Satisfiability}, \textsc{Independent Set}, \textsc{Dominating Set}, \textsc{Vertex-Coloring}, etc.), two \textit{feasible solutions} (e.g., satisfying truth assignments, independent sets, dominating sets, proper vertex-colorings, etc.) $S$ and $T$ are given along with a \textit{reconfiguration rule} that describes how to slightly modify one feasible solution to obtain a new one. 
The question is whether one can transform/reconfigure $S$ into $T$ via a sequence of feasible solutions such that each intermediate member is obtained from its predecessor by applying the given rule exactly once. Such a sequence, if it exists, is called a \textit{reconfiguration sequence}.

\subsubsection*{Distance Domination.}
In 1975, Meir and Moon~\cite{MeirM75} introduced the concept of a \textit{distance-$r$ dominating set} (which they called an ``$r$-covering''), combining the notions of ``distance'' and ``domination'' in graphs, where $r \geq 1$ is a fixed integer.
Given a graph $G$ and a fixed positive integer $r$, a \textit{distance-$r$ dominating set} (D$r$DS)—also known as an \textit{$r$-hop dominating set} or \textit{$r$-basis}—is a vertex subset $D$ such that every vertex of $G$ is within distance $r$ of at least one member of $D$. In particular, when $r = 1$, a D$1$DS is simply a \textit{dominating set} of $G$.

Given a graph $G$ and a positive integer $k$, the \textsc{Distance-$r$ Dominating Set} problem seeks to determine whether there exists a D$r$DS of $G$ with size at most $k$. 
It has been shown that the \textsc{Distance-$r$ Dominating Set} problem remains $\ttNP$-complete even when restricted to bipartite graphs or chordal graphs with diameter $2r + 1$~\cite{ChangN84}. 
Since Meir and Moon's initial work, the concept of distance domination in graphs has been extensively explored from various perspectives. 
For more information, readers are referred to Henning's survey~\cite{Henning20}.

\subsubsection*{Dominating Set Reconfiguration.}
Dominating set reconfiguration has been widely investigated in the literature from both algorithmic and graph-theoretical points of view. 
Consider placing a token on every vertex in a dominating set of a graph~$G$, ensuring that no vertex holds more than one token.
In the context of reconfiguring dominating sets, the following rules have been studied:
\begin{itemize}
	\item \textbf{Token Sliding ($\sfTS$):} A token may be relocated to one of its unoccupied neighboring vertices, as long as the new token set remains a dominating set.
	\item \textbf{Token Jumping ($\sfTJ$):} A token may be moved to any unoccupied vertex, provided that the resulting token set still forms a dominating set.
	\item \textbf{Token Addition/Removal ($\sfTAR(k)$):} A token can be added to or removed from the current token set, as long as the resulting set is a dominating set of size not exceeding a given threshold $k \geq 0$.
\end{itemize}
Given a graph~$G$, a fixed integer~$r \geq 1$, two dominating sets~$D_s$ and~$D_t$ of~$G$, and a reconfiguration rule~$\sfR \in \{\sfTS, \sfTJ, \sfTAR\}$, the \textsc{Dominating Set Reconfiguration (DSR)} problem under~$\sfR$ asks whether there exists a sequence of dominating sets transforming~$D_s$ into~$D_t$, where each intermediate set is obtained from the previous one by applying~$\sfR$ exactly once.
\ReviewRevise{We remark that if two dominating sets (or in general, any two vertex subsets) $D_s$ and $D_t$ can be transformed into each other under $\sfTS$ or $\sfTJ$, then they must have the same size.}

We briefly summarize some well-known algorithmic results and direct readers to~\cite{MynhardtN19} for recent progress from a graph-theoretic perspective.
Haddadan et al.~\cite{HaddadanIMNOST16} initially explored the computational complexity of \textsc{DSR} under~$\sfTAR$, demonstrating that the problem is~$\ttPSPACE$-complete for planar graphs with maximum degree six, bounded bandwidth graphs, split graphs, and bipartite graphs. 
On the positive side, they developed polynomial-time algorithms for solving \textsc{DSR} under~$\sfTAR$ on cographs, forests, and interval graphs. 
Bonamy et al.~\cite{BonamyDO21} observed that under certain constraints, the~$\sfTJ$ and~$\sfTAR$ rules are equivalent, and hence the aforementioned results hold for \textsc{DSR} under~$\sfTJ$ as well. 
Bonamy et al.~\cite{BonamyDO21} were the first to examine the computational complexity of \textsc{DSR} under~$\sfTS$, showing that the hardness results from Haddadan et al.~\cite{HaddadanIMNOST16} also apply under~$\sfTS$.
Conversely, Bonamy et al.~\cite{BonamyDO21} devised polynomial-time algorithms for solving \textsc{DSR} under~$\sfTS$ on cographs and dually chordal graphs (which include trees and interval graphs). 
Subsequently, Bousquet and Joffard~\cite{BousquetJ21} demonstrated that \textsc{DSR} under~$\sfTS$ is~$\ttPSPACE$-complete on circle graphs but solvable in polynomial time on circular-arc graphs, thereby addressing some open questions posed by Bonamy et al.~\cite{BonamyDO21}. 
More recently, K\v{r}i\v{s}t'an and Svoboda~\cite{KristanS23} strengthened the positive results of Bonamy et al.~\cite{BonamyDO21} by providing polynomial-time algorithms to find a \textit{shortest} reconfiguration sequence, if one exists, between two given dominating sets under~$\sfTS$, when the input graph is either a tree or an interval graph. 
However, their methods do not generalize to dually chordal graphs.

Mouawad et al.~\cite{MouawadN0SS17} initiated a systematic study of the parameterized complexity of various reconfiguration problems, including \textsc{DSR}. 
Two natural parameters are the number of tokens~$k$ and the length of a reconfiguration sequence~$\ell$. In~\cite{MouawadN0SS17}, Mouawad~et~al.\ proved that \textsc{DSR} under~$\sfTAR$ on general graphs is~$\mathtt{W[1]}$-hard when parameterized by~$k$, and~$\mathtt{W[2]}$-hard when parameterized by~$k + \ell$. 
Lokshtanov et al.~\cite{LokshtanovMPRS18}, considering graphs that exclude~$K_{d,d}$ as a subgraph for any fixed constant~$d$ (which includes bounded degeneracy and nowhere dense graphs), designed an~$\mathtt{FPT}$ algorithm for solving the problem when parameterized by~$k$. 
When parameterized solely by~$\ell$, it was noted in~\cite{BousquetMNS24} that the problem is fixed-parameter tractable on any graph class where first-order model-checking is fixed-parameter tractable. 
For further details, we refer readers to~\cite{LokshtanovMPRS18,BousquetMNS24} and the references therein.

\subsubsection*{Our Problem.}
One can define \textsc{Distance-$r$ Dominating Set Reconfiguration (D$r$DSR)} under $\sfR$ for any fixed integer $r \geq 2$, similarly to \textsc{DSR}. 
As far as we know, the classical complexity of \textsc{D$r$DSR} for $r \geq 2$ has not been explored yet. 
From the parameterized complexity perspective, Siebertz~\cite{Siebertz18} examined \textsc{D$r$DSR} under $\sfTAR$, parameterized by $k$, and demonstrated that for some fixed constant $r$, the problem is $\mathtt{W[2]}$-hard on somewhere dense graphs which are also closed under taking subgraphs. 
Conversely, Siebertz showed that the problem lies in $\mathtt{FPT}$ when restricted to nowhere dense graphs. 
From a graph-theoretical standpoint, DeVos et al.~\cite{DeVosDJS20} introduced the \textit{$\gamma_r$-graph} of a graph $G$—a reconfiguration graph whose nodes represent \textit{minimum} distance-$r$ dominating sets of $G$—and established several results concerning its realizability. 
In this work, we investigate \textsc{D$r$DSR} ($r \geq 2$) under $\sfTS$ and $\sfTJ$ from the classical complexity perspective, aiming to better understand the boundary between ``hard'' and ``easy'' instances on various graph classes.

\subsubsection*{Our Results.}  

\begin{figure}[t]
	\centering
	\includegraphics[width=0.9\textwidth]{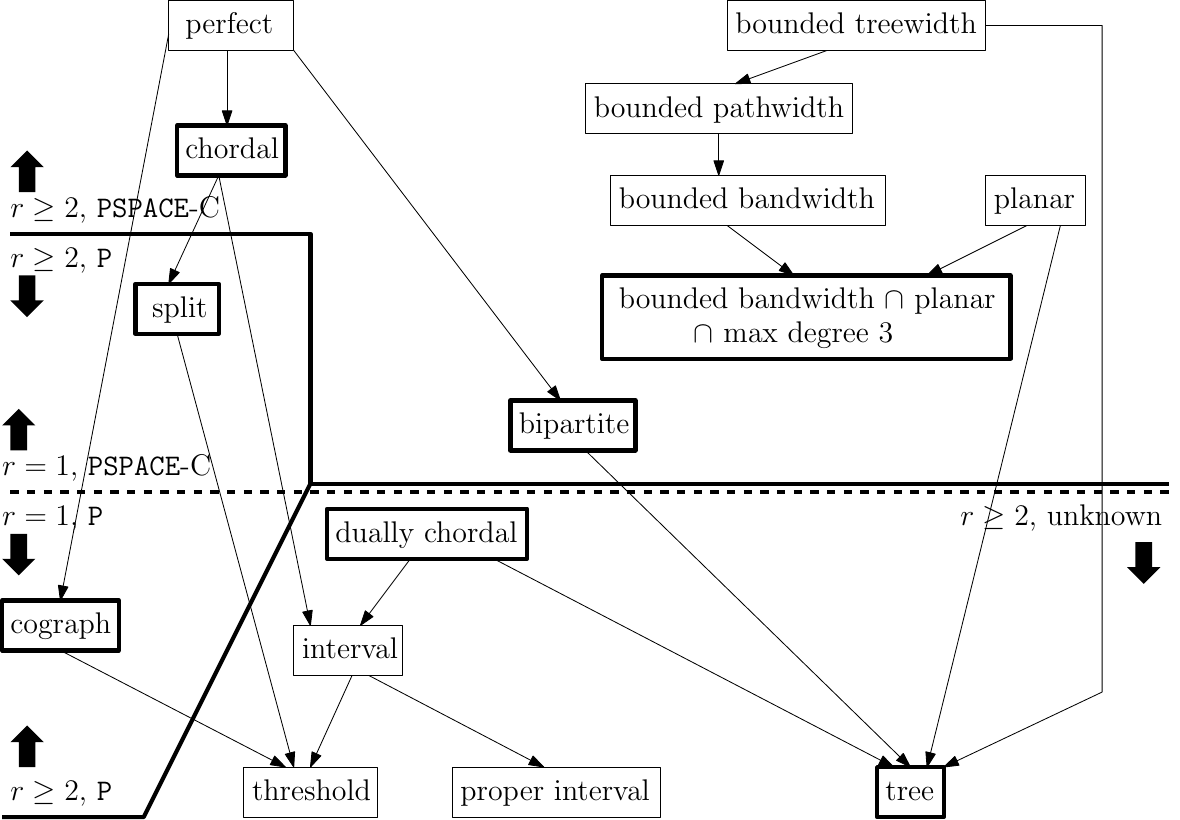}
	\caption{The complexity status of \textsc{D$r$DSR} for fixed $r \geq 1$ on different graph classes under $\sfTS$. Our results are for $r \geq 2$. Each arrow from graph class $A$ to graph class $B$ indicates that $B$ is properly included in $A$.}
	\label{fig:graph-classes-TS}
\end{figure}

\begin{figure}[t]
	\centering
	\includegraphics[width=0.9\textwidth]{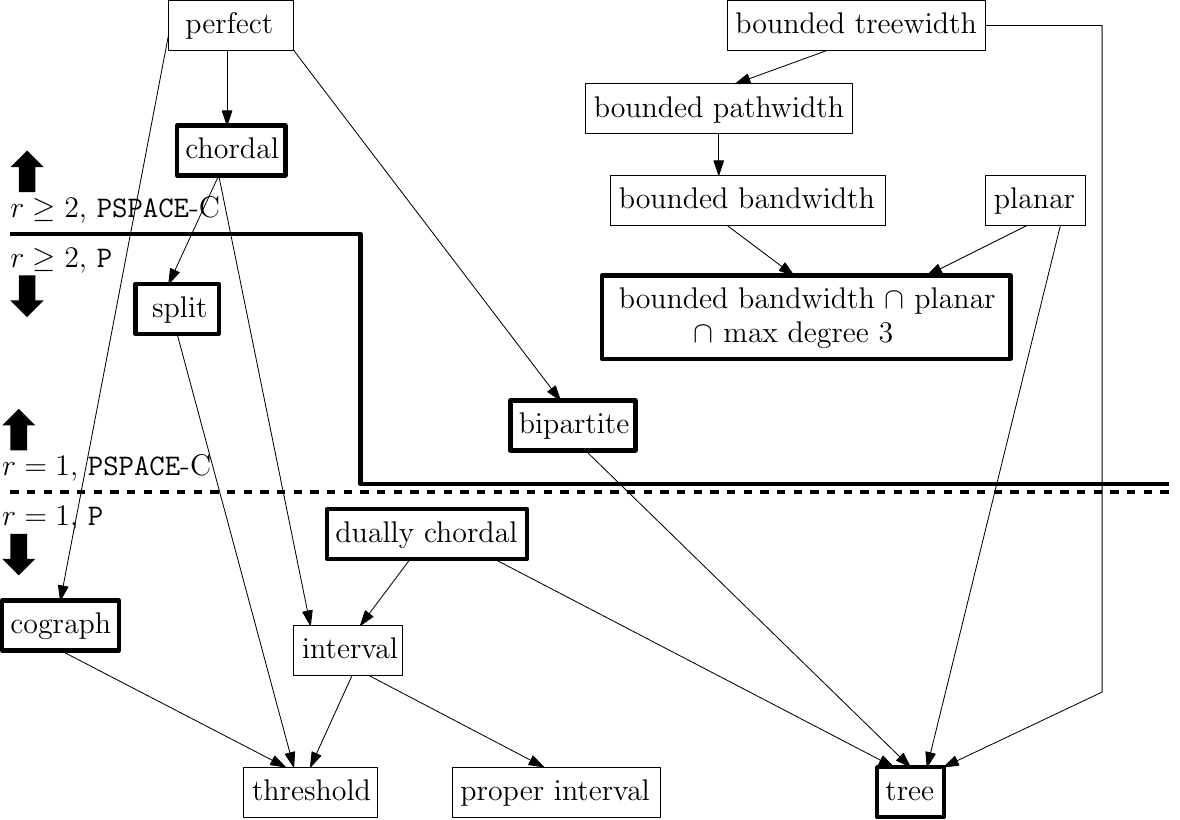}
	\caption{The complexity status of \textsc{D$r$DSR} for fixed $r \geq 1$ on different graph classes under $\sfTJ$. Our results are for $r \geq 2$. Each arrow from graph class $A$ to graph class $B$ indicates that $B$ is properly included in $A$.}
	\label{fig:graph-classes-TJ}
\end{figure}

In Section~\ref{sec:polytime}, we give polynomial-time algorithms for different graph classes. 
(See \figurename~\ref{fig:graph-classes-TS} and \figurename~\ref{fig:graph-classes-TJ}.)

\begin{itemize}
	\item In Section~\ref{sec:observe}, we give a few observations that, when combined with known results for \textsc{DSR}, lead to the polynomial-time solvability of \textsc{D$r$DSR} for $r \geq 2$ on dually chordal graphs (which include interval graphs and trees) under $\sfTJ$ ({\bf Proposition~\ref{prop:dually-chordal}}) and on cographs (i.e., $P_4$-free graphs) under both $\sfTS$ and $\sfTJ$ ({\bf Corollary~\ref{cor:TJ-cograph}}). 
	
	\item In Section~\ref{sec:split}, we show that \textsc{D$r$DSR} under either $\sfTS$ or $\sfTJ$ on split graphs can be solved in polynomial time for any fixed $r \geq 2$ ({\bf Theorem~\ref{thm:TS-split}} and {\bf Corollary~\ref{cor:TJ-split}}). 
	This result establishes a surprising dichotomy compared to the lower bound result on split graphs by Haddadan et al. \cite{HaddadanIMNOST16}, who show that the problem is $\ttPSPACE$-complete for $r=1$. 
	Additionally, we establish non-trivial bounds on the length of a \textit{shortest} reconfiguration sequence between two D$2$DSs of a split graph under both $\sfTS$ and $\sfTJ$, which may be of independent interest ({\bf Theorem~\ref{thm:split-shortest}}). 
	In particular, any $\sfTS$-sequence (resp., $\sfTJ$-sequence) has a natural lower bound on its length, 
	and we show that, on split graphs, one can design a $\sfTS$-sequence (resp., $\sfTJ$-sequence) between \textit{any} pair of D$2$DSs $D_s$ and $D_t$ that uses at most two (resp., one) more extra token-slides (resp., token-jump) than such lower bound. 
	Moreover, we also show that there exist instances $(G, D_s, D_t)$ on split graphs where one cannot use less than two extra moves under $\sfTS$ and less than one extra move under $\sfTJ$ to reconfigure $D_s$ into $D_t$ and vice versa.

	\item Our result on dually chordal graphs above implies that D$r$DSR ($r \geq 2$) under $\sfTJ$ on trees (and indeed on dually chordal graphs) can be solved in quadratic time. 
	In Section~\ref{sec:trees}, we further improve this running time by designing a linear-time algorithm for solving the problem on trees ({\bf Theorem~\ref{thm:TJ-trees}}).	
\end{itemize}

In Section~\ref{sec:hardness}, we give our hardness results. 

\begin{itemize}
	\item In Section~\ref{sec:planar}, we show that for any fixed $r \geq 1$, \textsc{D$r$DSR} under either $\sfTS$ or $\sfTJ$ remains $\ttPSPACE$-complete even on planar graphs of maximum degree three and bounded bandwidth ({\bf Theorem~\ref{thm:planar-maxdeg3}}). 
    Our result on planar graphs improves the previously known results of Haddadan et al.~\cite{HaddadanIMNOST16} and Bonamy et al.~\cite{BonamyDO21}, which are only for $r = 1$ and planar graphs of maximum degree six and bounded bandwidth.
    The key for our improvement is our non-trivial polynomial-time reduction from the well-known \textsc{NCL} problem~\cite{HearnD05}, while in the previously known results, the natural choice of problem for reduction is \textsc{Vertex Cover Reconfiguration}.
	
	\item In Sections~\ref{sec:chordal} and~\ref{sec:bipartite}, we also show that some known hardness results for \textsc{DSR} can be extended to \textsc{D$r$DSR} when $r \geq 2$. 
	More precisely, we show that under either $\sfTS$ or $\sfTJ$, \textsc{D$r$DSR} ($r \geq 2$) is $\ttPSPACE$-complete on chordal graphs ({\bf Theorem~\ref{thm:chordal}}) and on bipartite graphs ({\bf Theorem~\ref{thm:bipartite}}).
\end{itemize}

\section{Preliminaries}
\label{sec:preliminaries}

\subsubsection*{Graph Notation.}
For concepts and notations not defined here, we refer readers to~\cite{Diestel2017}. 
Unless otherwise mentioned, throughout this paper, we consider simple, connected\footnote{\ReviewRevise{It suffices to assume that the input graph is connected. Otherwise, each component can be handled separately.}}, undirected graphs $G$ with \ReviewRevise{vertex set} $V(G)$ and \ReviewRevise{edge set} $E(G)$.
For any pair of vertices $u, v$, the \textit{distance} between $u$ and $v$ in $G$, denoted by $\dist_G(u, v)$, is the length of the shortest path between them. 
The \textit{diameter} of $G$, denoted by $\diam(G)$, is the largest distance between any pair of vertices.
For two vertex subsets $X, Y$, we use $X - Y$ and $X + Y$ to indicate $X \setminus Y$ and $X \cup Y$, respectively. 
If $Y$ contains a single vertex $u$, we write $X - u$ and $X + u$ instead of $X - \{u\}$ and $X + \{u\}$, respectively. 
We denote by $X \Delta Y$ their \textit{symmetric difference}, i.e., $X \Delta Y = (X - Y) + (Y - X)$. 
For a subgraph $H$ of $G$, we denote $G - H$ as the graph obtained from $G$ by deleting all vertices of $H$ and their incident edges in $G$.

A \textit{dominating set (DS)} of $G$ is a vertex subset $D$ such that for every $u \in V(G)$, there exists $v \in D$ such that $\dist_G(u, v) \leq 1$. 
For a fixed positive integer $r$, a \textit{distance-$r$ dominating set (D$r$DS)} of $G$ is a vertex subset $D$ such that for every $v \in V(G)$, there exists $v \in D$ such that $\dist_G(u, v) \leq r$. 
In particular, any D$1$DS is also a DS, and vice versa. Let $N^r_G[u]$ be the set of all vertices of distance at most $r$ from $u$ in $G$. 
We say that a vertex $v$ is \textit{$r$-dominated} by $u$ (or $u$ \textit{$r$-dominates} $v$) if $v \in N^r_G[u]$. 
We say that a vertex subset $X$ is \textit{$r$-dominated} by some vertex subset $Y$ if each vertex in $X$ is $r$-dominated by some vertex in $Y$. 
A D$r$DS is nothing but a vertex subset $D$ that $r$-dominates $V(G)$. We denote by $\gamma_r(G)$ the size of a minimum D$r$DS of $G$.

\subsubsection*{Graph Classes.}
We mention some graph classes that will be considered in this paper. 
A graph $G = (V, E)$ is called a \textit{split graph} if its vertex set $V$ is partitioned into two subsets, $K$ and $S$, which respectively induce a clique and an independent set of $G$. 
For any split graph $G$, we assume that such a partition is always given and write $G = (K \uplus S, E)$. 
A graph $G$ is called a \textit{chordal graph} if every cycle of $G$ on four or more vertices has a chord—an edge joining two non-adjacent vertices in the cycle. 
It is known that any split graph is also a chordal graph. 
A \textit{bipartite graph} $G$ is a graph where $V(G)$ can be partitioned into two disjoint vertex subsets $X$ and $Y$ where no edge of $G$ joins two vertices in $X$ or two vertices in $Y$. 
$G$ is a \textit{planar graph} if it can be drawn in the plane in such a way that no edges cross each other. 
A \textit{tree} is a graph having no cycles. 
The \textit{bandwidth} of a graph $G$ is the minimum over all injective assignments $f: V(G) \to \mathbb{N}$ of $\max_{uv \in E(G)}|f(u) - f(v)|$. 
$G$ is called a \textit{bounded bandwidth graph} if there exists a constant $b$ such that the bandwidth of $G$ is at most $b$.

\subsubsection*{Reconfiguration Notation.}
Throughout this paper, we write ``$(G, D_s, D_t)$ under $\sfR$'' to indicate an instance of \textsc{D$r$DSR} 
where $D_s$ and $D_t$ are two given D$r$DSs of a graph $G$, and the reconfiguration rule is $\sfR \in \{\sfTS, \sfTJ\}$.
Imagine that a token is placed on each vertex in a D$r$DS of a graph $G$. 
A \textit{$\sfTS$-sequence} in $G$ between two D$r$DSs $D_s$ and $D_t$ is a sequence $\mathcal{S} = \langle D_s = D_0, D_1, \dots, D_q = D_t \rangle$ such that for $i \in \{0, \dots, q-1\}$, the set $D_i$ is a D$r$DS of $G$, and there exists a pair $x_i, y_i \in V(G)$ such that $D_i - D_{i+1} = \{x_i\}$, $D_{i+1} - D_i = \{y_i\}$, and $x_iy_i \in E(G)$.
A \textit{$\sfTJ$-sequence} in $G$ is defined similarly, but without the restriction $x_iy_i \in E(G)$.
For a given rule $\sfR \in \{\sfTS, \sfTJ\}$, we say that $D_{i+1}$ is obtained from $D_i$ by \textit{immediately sliding/jumping} a token from $x_i$ to $y_i$ under rule $\sfR$, and we write $x_i \reconf[\sfR]{G} y_i$.
This allows us to express $\mathcal{S}$ as a sequence of token-moves: $\mathcal{S} = \langle x_0 \reconf[\sfR]{G} y_0, \dots, x_{q-1} \reconf[\sfR]{G} y_{q-1} \rangle$.
In short, $\mathcal{S}$ can be viewed as an (ordered) sequence of either D$r$DSs or token-moves.
(Recall that we defined $\mathcal{S}$ as a sequence between $D_s$ and $D_t$. As a result, when regarding $\mathcal{S}$ as a sequence of token-moves, we implicitly assume that the initial D$r$DS is $D_s$.) 
With respect to the latter viewpoint, we say that $\mathcal{S}$ \textit{slides/jumps a token $t$ from $u$ to $v$ in $G$} if $t$ is initially placed on $u \in D_0$ and finally on $v \in D_q$ after performing $\mathcal{S}$. 
The \textit{length} of a $\sfR$-sequence is simply the number of times the rule $\sfR$ is applied.

\section{Polynomial-Time Algorithms}
\label{sec:polytime}

\subsection{Observations}
\label{sec:observe}

\subsubsection*{Graphs and Their Powers.}
For a graph $G$ and an integer $s \geq 1$, the \textit{$s^{th}$ power of $G$} is the graph $G^s$ whose vertices are $V(G)$ and two vertices $u, v$ are adjacent in $G^s$ if $\dist_G(u, v) \leq s$. 
Observe that $D$ is a D$r$DS of $G$ if and only if $D$ is a DS of $G^{r}$.
The following proposition is straightforward.
\begin{proposition}\label{prop:power-graph}
	Let $\calG$ and $\calH$ be two graph classes and suppose that for every $G \in \calG$ we have $G^r \in \calH$ for some fixed integer $r \geq 1$.
	If \textsc{DSR} under $\sfTJ$ on $\calH$ can be solved in polynomial time, so does \textsc{D$r$DSR} under $\sfTJ$ on $\calG$. %
\end{proposition}
\begin{proof}
	Since $D$ is a D$r$DS of $G$ if and only if $D$ is a DS of $G^{r}$, any $\sfTJ$-sequence in $G$ between two D$r$DSs can be converted to a $\sfTJ$-sequence in $G^r$ between two corresponding DSs and vice versa.
\end{proof}

\subsubsection*{Dually Chordal Graphs.}
We refer readers to \cite{BrandstadtDCV98} for a precise definition of this graph class.
It is known that \textsc{DSR} can be solved in quadratic time under $\sfTS$ on connected dually chordal graphs \cite[Theorem 13]{BonamyDO21}.
Specifically, given an instance $(G, D_s, D_t)$ of \textsc{DSR} under $\sfTS$, where $D_s$ and $D_t$ are two dominating sets of a connected dually chordal graph $G$, one can construct in quadratic time a $\sfTS$-sequence between $D_s$ and $D_t$.
Since any $\sfTS$-sequence is also a $\sfTJ$-sequence, the result also applies for $\sfTJ$.
In other words, \textsc{DSR} under $\sfTJ$ on connected dually chordal graphs can be solved in quadratic time.

One can construct in quadratic time a $\sfTJ$-sequence between $D_s$ and $D_t$ on disconnected dually chordal graphs $G$ as follows.
Bonamy et al. also proved that one can construct in linear time a canonical minimum dominating set $C$ and respectively reconfigure $D_s$ and $D_t$ under $\sfTS$ into $D_s^\prime$ and $D_t^\prime$ such that $C \subseteq D_s^\prime \cap D_t^\prime$ \cite[Lemma 12]{BonamyDO21}.
Moreover, in a component, note that any token not in a canonical minimum dominating set can be arbitrarily moved/jumped to any other component.
As a result, one can indeed reconfigure $D_s$ and $D_t$ under $\sfTJ$ such that, finally, their intersections with each component are of equal size.
Now, it remains to solve for each component independently, which can be done as described in \cite[Theorem 13]{BonamyDO21}.
In summary, we have the following proposition.

\begin{proposition}\label{prop:dually-chordal}
\textsc{DSR} under $\sfTJ$ is quadratic-time solvable on dually chordal graphs.
\end{proposition}

It is also known that the power of any dually chordal graph is also a dually chordal graph \cite{BrandstadtDCV98}.
Additionally, any tree or interval graph is also dually chordal \cite{BonamyDO21}.
Along with Proposition \ref{prop:power-graph}, we obtain the following corollary.

\begin{corollary}\label{cor:dually-chordal}
\textsc{D$r$DSR} under $\sfTJ$ on dually chordal graphs is quadratic-time solvable for any $r \geq 2$.
Consequently, the result also applies to trees and interval graphs.
\end{corollary}

\subsubsection*{Graphs With Bounded Diameter Components.}
We prove the following observation for graphs satisfying that each component's diameter is upper bounded by the same constant.
\begin{proposition}\label{prop:bounded-diameter}
	Let $G$ be any graph such that there is some constant $c > 0$ satisfying $\diam(C_G) \leq c$ for any component $C_G$ of $G$.
	Then, \textsc{D$r$DSR} on $G$ under $\sfR \in \{\sfTS, \sfTJ\}$ is in $\ttP$ for every $r \geq c$.
\end{proposition}
\begin{proof}
	When $r \geq c$, any size-$1$ vertex subset of $G$ is a D$r$DS.
	In this case, observe that any token-jump (and therefore token-slide) from one vertex to any unoccupied vertex always results a new D$r$DS.
	Thus, under $\sfTJ$, any instance of \textsc{D$r$DSR} is a yes-instance.  
	On the other hand, under $\sfTS$, the answer depends on the number of tokens in each component.
	More precisely, for any instance $(G, D_s, D_t)$ of \textsc{D$r$DSR} under $\sfTS$, if $|D_s \cap V(C_G)| = |D_t \cap V(C_G)|$ for any component $C_G$ of $G$ then it is a yes-instance; otherwise it is a no-instance.
\end{proof}

Since any connected cograph ($P_4$-free graph) has diameter at most $2$, the following corollary is straightforward.
\begin{corollary}\label{cor:TJ-cograph}
	\textsc{D$r$DSR} under $\sfR \in \{\sfTS, \sfTJ\}$ on cographs is in $\ttP$ for any $r \geq 2$.
\end{corollary}

\subsection{Split Graphs}
\label{sec:split}

In this section, we consider split graphs. 
We begin by proving some useful properties when moving tokens in a split graph.
\begin{lemma}\label{lem:TS-split}
	Let $D$ be a D$2$DS of a split graph $G = (K \uplus S, E)$.
	\begin{enumerate}[(a)]
		\item For every pair $u \in D \cap K$ and $v \in K - D$, the set $D - u + v$ is a D$2$DS of $G$.
		\item For every pair $u \in D \cap S$ and $v \in K - D$, the set $D - u + v$ is a D$2$DS of $G$.
		\item For every pair $u \in D \cap K$ and $v \in S - D$, the set $D - u + v$ is a D$2$DS of $G$ if $(D \cap K) - u \neq \emptyset$.
	\end{enumerate}
\end{lemma}
\begin{proof}
	By definition, $\dist_G(v, w) \leq 2$ for any $w \in V(G)$.
	Consequently, in both (a) and (b), $\{v\}$ is a D$2$DS of $G$, and therefore so is $D - u + v \supseteq \{v\}$.
	In (c), since $(D \cap K) - u \neq \emptyset$, there must be a vertex $x \in D \cap K$ such that $x \neq u$.
	Again, since $\{x\}$ is a D$2$DS of $G$, so is $D - u + v \supseteq \{x\}$.
\end{proof}

A direct consequence of Lemma~\ref{lem:TS-split} is as follows.
\begin{corollary} \label{cor:TSfree}
    Let $G = (K \uplus S, E)$ be a split graph and $D$ be any D$2$DS of $G$.
    For any $u \in D$ and $v \in N_G(u) - D$, if $(D - u) \cap K \neq \emptyset$ then $D - u + v$ is always a D$2$DS of $G$.
\end{corollary}

In the following theorem, we prove that \textsc{D$r$DSR} under $\sfTS$ on split graphs is polynomial-time solvable.
\begin{theorem}\label{thm:TS-split}
	\textsc{D$r$DSR} under $\sfTS$ on split graphs is in $\ttP$ for any $r \geq 2$.
	In particular, when $r = 2$, for any pair of size-$k$ D$2$DSs $D_s, D_t$ of a split graph $G = (K \uplus S, E)$, there is a $\mathsf{TS}$-sequence in $G$ between $D_s$ and $D_t$.
\end{theorem}

  \begin{proof}
	Proposition~\ref{prop:bounded-diameter} settles the case $r \geq 3$. 
	It remains to consider the case $r = 2$. 
	We claim that for any pair of size-$k$ D$2$DSs $D_s$ and $D_t$ of a split graph $G = (K \uplus S, E)$, there exists a $\mathsf{TS}$-sequence in $G$ between $D_s$ and $D_t$.

	Suppose that $p = |K| \geq 1$ and $q = |S| \geq 1$. 
	If $p = 1$, then $G$ is a star graph $K_{1, q}$, and any size-$1$ vertex subset of $G$ is also a D$2$DS. 
	In this case, any token can move freely, and a $\sfTS$-sequence between $D_s$ and $D_t$ can be easily designed.

	Now, suppose that $p \geq 2$. 
	In this case, observe that as long as there is a token in $K$, Corollary~\ref{cor:TSfree} ensures that any other token can move freely. 
	Using this observation, we design a $\sfTS$-sequence between $D_s$ and $D_t$ as follows.

	\begin{itemize}
	\item If there is no token in $K$, move an arbitrary token from some vertex $u \in D_s \cap S$ to one of its neighbors $v \in K$. (As $G$ is connected, such a vertex $v$ exists.)
	
	\item While keeping a token in $K$, move tokens to vertices in $S \cap D_t$ as follows: Pick a $v \in S \cap D_t$, choose a closest token to $v$ on some vertex $v_{\text{close}}$, and move it to $v$.
	If $v_{\text{close}}$ is the only vertex with a token in $K$ at that time and $v$ is not the last vertex in $D_t$ that needs a token, first move another token from $S$ to $K$. 
	This can be done either by shifting the token on $v_{\text{close}}$ to another vertex in $K$ and then shifting a token to $v_{\text{close}}$, or by shifting a token to a vertex in $K \setminus v_{\text{close}}$. 
	If $v$ is the last vertex in $D_t$ that needs a token, simply move the token on $v_{\text{close}}$ to $v$ along one of their shortest paths $P$. 
	Note that all vertices of $P$ except $v$ are in $K$, and $D_t$ is a D$2$DS. Thus, every move along $P$ results in a new D$2$DS.

	\item Finally, move any remaining tokens (if they exist) to $K \cap D_t$.
	\end{itemize}
 
        \end{proof}

Since any $\sfTS$-sequence in $G$ is also a $\sfTJ$-sequence, a direct consequence of Theorem~\ref{thm:TS-split} and Proposition~\ref{prop:bounded-diameter} is as follows.
\begin{corollary}\label{cor:TJ-split}
	\textsc{D$r$DSR} under $\sfTJ$ on split graphs is in $\ttP$ for any $r \geq 2$.
\end{corollary}

We now consider \emph{shortest} reconfiguration sequences in split graphs.
Observe that each $\sfR$-sequence ($\sfR \in \{\sfTS, \sfTJ\}$) between two D$r$DSs $D_s$ and $D_t$ induces a bijection $f$ between them: the token on $u \in D_s$ must ultimately be placed on $f(u) \in D_t$, and vice versa.
We denote the length of a \emph{shortest} $\sfR$-sequence in $G$ between two D$r$DSs $D_s$ and $D_t$ by $\opt_{\sfR}(G, D_s, D_t)$. 
If there is no $\sfR$-sequence between $D_s$ and $D_t$, we define $\opt_{\sfR}(G, D_s, D_t) = \infty$.

Let $M^\star_{\sfTS}(G, D_s, D_t) = \min_{f} \sum_{i=1}^{k} \dist_G(s_i,f(s_i))$, where $f$ is a bijection between vertices of $D_s$ and $D_t$, and let $\displaystyle M^\star_{\sfTJ}(G, D_s, D_t) = \frac{|D_s \Delta D_t|}{2}$. 
Observe that $\opt_{\sfTS}(G, D_s, D_t) \geq M^\star_{\sfTS}(G, D_s, D_t)$; in order to slide a token from $s_i \in D_s$ to $f(s_i) \in D_t$ for some $i \in \{1, \dots, k\}$, one cannot use fewer than $\dist_G(s_i, f(s_i))$ token-slides.
Additionally, we also have $\opt_{\sfTJ}(G, D_s, D_t) \geq M^\star_{\sfTJ}(G, D_s, D_t)$; each token in $D_s \setminus D_t$ may jump directly to some vertex in $D_t \setminus D_s$.

Furthermore, observe that for any instance where both $D_s$ and $D_t$ are subsets of $K$, one can indeed construct $\sfTS$- and $\sfTJ$-sequences that are, respectively, of lengths $M^\star_{\sfTS}(G, D_s, D_t)$ and $M^\star_{\sfTJ}(G, D_s, D_t)$. (Recall that any vertex in $K$ forms a size-$1$ D$2$DS of $G$.)

In~\cite[Theorem 4.3]{Heuvel13}, van~den~Heuvel proved the following theorem.

\begin{theorem}[\cite{Heuvel13}]\label{thm:Heuvel13}
For any graph $G$ and two token-sets (i.e., vertex subsets of $V(G)$) $U$ and $V$, there exists a shortest $\sfTS$-sequence between $U$ and $V$ of length exactly $M^\star_{\sfTS}(G, U, V)$.
\end{theorem}

Indeed, we are not interested in the case $r \geq 3$. 
In this case, any token-jump or token-slide results in a new D$r$DS of $G$. 
As a result, it is not hard to construct a shortest reconfiguration sequence: 
Under $\sfTS$, one can apply Theorem~\ref{thm:Heuvel13}. 
Under $\sfTJ$, just jump tokens one by one from $D_s \setminus D_t$ to $D_t \setminus D_s$.

The rest of this section is devoted to proving the following theorem when $r = 2$.

\begin{theorem}\label{thm:split-shortest}
	For any instance $(G, D_s, D_t)$ of \textsc{D$2$DSR} on split graphs,
	\begin{enumerate}[(a)]
		\item $M^\star_{\sfTS}(G, D_s, D_t) \leq \opt_{\sfTS}(G, D_s, D_t) \leq M^\star_{\sfTS}(G, D_s, D_t) + 2$. 
		\item $M^\star_{\sfTJ}(G, D_s, D_t) \leq \opt_{\sfTJ}(G, D_s, D_t) \leq M^\star_{\sfTJ}(G, D_s, D_t) + 1$.  
	\end{enumerate}
	Additionally, there are instances where $\opt_{\sfTS}(G, D_s, D_t) = M^\star_{\sfTS}(G, D_s, D_t) + 2$ and instances where $\opt_{\sfTJ}(G, D_s, D_t) = M^\star_{\sfTJ}(G, D_s, D_t) + 1$.
\end{theorem}

We first consider the $\sfTS$ rule.
The following lemma is also a direct consequence of Theorem~\ref{thm:Heuvel13}.

\begin{lemma}\label{lem:alwaysTSsequence}
Let $G = (K \uplus S, E)$ be a split graph. 
Suppose that \textit{any} sequence of token-slides in $G$ is a $\sfTS$-sequence, meaning that every time a token is slid along an edge, the resulting token-set is always a D$2$DS. 
Then, for any pair of D$2$DSs $D_s$ and $D_t$ of $G$, there exists a $\sfTS$-sequence of length exactly $M^\star_{\sfTS}(G, D_s, D_t)$.
\end{lemma}

\begin{lemma}\label{lem:TS-split-shortest}
	There exists a split graph $G = (K \uplus S, E)$ and two size-$k$ D$2$DSs $D_s, D_t$ of $G$ such that $\displaystyle\opt_{\sfTS}(G, D_s, D_t) = M^\star_{\sfTS}(G, D_s, D_t) + 2$, for $k \geq 2$.
\end{lemma}
\begin{proof}

We construct a split graph $G = (K \uplus S, E)$ as follows (see Figure \ref{fig:TS-split}):

\begin{itemize}
\item The set $K$ contains 2 vertices, labeled $s_1 = t_1$ and $a$. Vertices in $K$ form a clique in $G$.
\item The set $S$ contains $2k - 1$ vertices, labeled $s_2, \dots, s_k, t_2, \dots, t_k, b$. Vertices in $S$ form an independent set in $G$.
\item We join $s_1 = t_1$ to every vertex in $\{s_2, \dots, s_k, t_2, \dots, t_k\}$ and join $a$ to $b$.
\end{itemize}

Let $D_s = \{s_1, \dots, s_k\}$ and $D_t = \{t_1, \dots, t_k\}$. Since $s_1 = t_1 \in K$, both $D_s$ and $D_t$ are D$2$DSs of $G$.

	\begin{figure}[ht]
		\centering
		\includegraphics[width=0.5\textwidth]{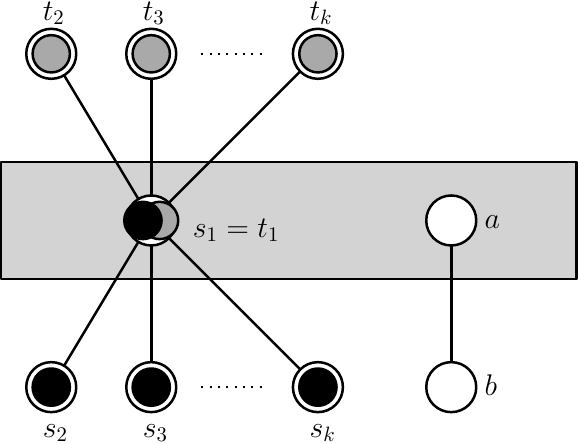}
		\caption{Construction of a split graph $G = (K \uplus S, E)$ satisfying Lemma~\ref{lem:TS-split-shortest}. Vertices in the light gray box are in $K$. Tokens in $D_s$ and $D_t$ are respectively marked by black and gray circles.}
		\label{fig:TS-split}
	\end{figure}

Let $M^\star = M^\star_{\sfTS}(G, D_s, D_t)$. 
One can verify that $M^\star = \sum_{i=1}^{k} \dist_G(s_i, t_i) = 2k - 2$. 
Additionally, any $\sfTS$-sequence $\calS$ in $G$ between $D_s$ and $D_t$ must begin with sliding the token on $s_1$ to one of its unoccupied neighbors. 
A $\sfTS$-sequence of length exactly $M^\star$ would require sliding the token on $s_1$ to one of $t_2, \dots, t_k$. 
However, this is not possible, as the vertex $b$ would not be $2$-dominated by any token in the resulting token-set.

Thus, the only viable initial move is to slide the token on $s_1$ to $a$. 
Next, one can slide each token on $s_i$ ($2 \leq i \leq k$) to $t_i$ along the $s_it_i$-path (of length $2$) in $G$. 
Finally, the token on $a$ is slid back to $t_1 = s_1$. 
This completes our construction of a $\sfTS$-sequence of length $2(k-1) + 2 = M^\star + 2$ between $D_s$ and $D_t$.

Observe that the moves from $s_1$ to $a$ and then later from $a$ to $s_1$ are essential and cannot be omitted. 
Furthermore, it requires at least two token-slides to transform a token on any $s_i$ to some $t_j$ for $2 \leq i, j \leq k$. 
Hence, in total, any $\sfTS$-sequence requires at least $M^\star + 2$ token-slides. 
Therefore, our constructed $\sfTS$-sequence is the shortest possible.

\end{proof}

Let $(G, D_s, D_t)$ be any instance of \textsc{D$2$DSR} under $\sfTS$, and let $f: D_s \to D_t$ be any bijection between the vertices of $D_s$ and $D_t$. For each $i \in \{1, \dots, k\}$, we call $f(s_i)$ the \textit{target} of the token on $s_i$. 
By \textit{handling the token on $s_i$}, we mean constructing a $\sfTS$-sequence that moves the token on $s_i$ to its target $f(s_i)$.
By \textit{swapping the targets} of $s_i$ and $s_j$, we mean applying the following procedure: Define a new bijection $g: D_s \to D_t$ such that $g(s_\ell) = f(s_\ell)$ for every $\ell \in \{1, \dots, k\} \setminus \{i, j\}$, $g(s_i) = f(s_j)$, and $g(s_j) = f(s_i)$. Then, reassign $f \gets g$.

For any instance $(G, D_s, D_t)$ of \textsc{D$2$DSR} under $\sfTS$ on split graphs, we now prove an upper bound of $\opt_{\sfTS}(G, D_s, D_t)$.
\begin{lemma}\label{lem:TS-split-shortest-bound}
	For any instance $(G, D_s, D_t)$ of \textsc{D$2$DSR} under $\sfTS$ on split graphs, $\opt_{\sfTS}(G, D_s, D_t) \leq M^\star_{\sfTS}(G, D_s, D_t) + 2$.
\end{lemma}
\begin{proof}
	Let $M^\star = M^\star_{\sfTS}(G, D_s, D_t)$.
	We describe an algorithm to construct a $\sfTS$-sequence of length at most $M^\star + 2$ between any pair of D$2$DSs, $D_s$ and $D_t$ of $G = (K \uplus S, E)$.
	Take any bijection $f: D_s \to D_t$ such that $M^\star = \sum_{i=1}^{k}\dist_G(s_i, f(s_i))$, where $D_s = \{s_1, \dots, s_k\}$.

Intuitively, we plan to move each token on $s_i$ to $f(s_i)$ along some shortest $s_i f(s_i)$-path. When there is an ``obstacle token'' along the path, we try the \textit{swapping targets} procedure. 
If the \textit{swapping targets} procedure does not work, we attempt to slide the ``obstacle'' token one step away, handling every other token that passes through the ``obstacle,'' and then move the ``obstacle token'' to its target. 
We claim that this entire procedure can be completed with at most $M^\star + 2$ token-slides.

Observe that Lemma~\ref{lem:alwaysTSsequence} addresses the case when $G$ has a diameter of at most $2$. 
Thus, we only need to consider the case when $G$ has a diameter of $3$. 
Since $G$ has a diameter of $3$, we have $|K| \geq 2$. We can further assume that $D_s \cap K \neq \emptyset$; otherwise, pick any $s_i \in D_s \cap S$ and slide the token on $s_i$ one step to one of its neighbors $u \in K$ in a shortest $s_i f(s_i)$-path in $G$. 
(As $G$ is connected, such a vertex $u$ exists.) 
Recall that Corollary~\ref{cor:TSfree} implies that the resulting token-set after each such move is a D$2$DS of $G$. 
Let $s_i$ be an arbitrary vertex in $D_s \cap K$. 
As $s_i \in K$, note that $s_i$ by itself is a valid D$2$DS of the graph.

We emphasize that from the proof of~\cite[Theorem 4.3]{Heuvel13}, it follows that every time Lemma~\ref{lem:alwaysTSsequence} is applied, moving a token\footnote{This token may or may not be from $s_j$, as sometimes you need to ``swap targets'' when having an ``obstacle token''.} to $f(s_j)$ always takes \textit{exactly} $\dist_G(s_j, f(s_j))$ token-slides, regardless of how the token is moved. This is the main reason why we can consider the cases separately as below and finally combine them to have a $\sfTS$-sequence with at most $M^\star + 2$ token-slides.

Observe that we can \textit{handle any token} on $s_j$, $j \neq i$, such that there exists a shortest $s_jf(s_j)$-path in $G$ that does not contain $s_i$. The reason is that, in this case, the token on $s_i$ is never moved when \textit{handling these tokens}. Thus, from the observation of Corollary~\ref{cor:TSfree}, we can apply Lemma~\ref{lem:alwaysTSsequence}.

It remains to handle every token $s_j$, $j \neq i$, where all shortest $s_jf(s_j)$-paths in $G$ pass through $s_i$. Let $\mathcal{A}$ be the set of all such tokens. 
We partition $\mathcal{A}$ into two disjoint subsets $\mathcal{A}_1$ and $\mathcal{A}_2$. $\mathcal{A}_1$ consists of all $s_j \in \mathcal{A}$ such that there exists a shortest $s_jf(s_j)$-path in $G$ that does not contain $f(s_i)$.
$\mathcal{A}_2$ consists of all $s_j \in \mathcal{A}$ such that every shortest $s_jf(s_j)$-path in $G$ contains $f(s_i)$. 
Intuitively, $\mathcal{A}_1$ consists of every $s_j$ where one can move the token on $s_j$ along some shortest $s_jf(s_j)$-path without possibly having a token on $f(s_i)$ as an ``obstacle token'', while $\mathcal{A}_2$ consists of every $s_j$ where one cannot. 
Observe that since $s_i \in K$, we have $0 \leq \dist_G(s_i, f(s_i)) \leq 2$. 
Moreover, note that if $f(s_i) \in S$, the vertex $f(s_i)$ has no token; 
otherwise, there must be some $s_j \in \mathcal{A}$ such that $s_j = f(s_i)$, 
and swapping the targets of $s_i$ and $s_j$ would result in a new bijection $f$ where $\sum_{i=1}^k\dist_G(s_i, f(s_i))$ is smaller than $M^\star$, which contradicts our assumption. 
(Note that any shortest $s_jf(s_j)$-path where $s_j \in \mathcal{A}$ contains $s_i$.)

We consider the following cases:

\begin{itemize}
        
 \item \textbf{Case~A: $\dist_G(s_i, f(s_i)) = 2$.} In this case, observe that $\mathcal{A}_2 = \emptyset$, and therefore, it suffices to show how to handle tokens on vertices in $\mathcal{A}_1$.
    Let $P_i$ be a shortest $s_if(s_i)$-path in $G$, and let $u$ be the common neighbor of $s_i$ and $f(s_i)$ in $P_i$.
    Note that $u \in K$ and $f(s_i) \in S$.
    
    \begin{itemize}
        \item Let us first consider the case where $u$ has no token.
        Again, $\mathcal{A}_1$ can be partitioned into two disjoint subsets $\mathcal{A}_{11}$ and $\mathcal{A}_{12}$, where $\mathcal{A}_{11}$ consists of all $s_j \in \mathcal{A}_1$ such that there exists a shortest $s_jf(s_j)$-path in $G$ that does not contain $u$, and $\mathcal{A}_{12}$ consists of all $s_j \in \mathcal{A}_1$ such that every shortest $s_jf(s_j)$-path in $G$ contains $u$.
        
        To handle tokens on vertices in both $\mathcal{A}_{11}$ and $\mathcal{A}_{12}$, the first step is to directly slide the token on $s_i$ to $u$, which can be done since $u \in K$.
        At this point, all tokens on $s_j \in \mathcal{A}_{11}$ can be \textit{handled} via shortest $s_jf(s_j)$-paths that do not contain $u$ using Lemma~\ref{lem:alwaysTSsequence}, as $s_i$ has no token, and this \textit{handling} process does not involve $u \in K$. 
        
        Observe that for all $s_j \in \mathcal{A}_{12}$, $s_js_iuf(s_j)$ must form a shortest $s_jf(s_j)$-path in $G$; otherwise (i.e., $s_jus_if(s_j)$ forms a shortest $s_jf(s_j)$-path in $G$), by swapping the targets of $s_i$ and $s_j$, we obtain a bijection with a total sum $\sum_{i = 1}^k\dist_G(s_i, f(s_i))$ smaller than $M^\star$, which contradicts our assumption.
        (Note that in the latter case, $\dist_G(s_i, f(s_j)) + \dist_G(s_j, f(s_i)) = 1 + 2 = 3 < \dist_G(s_i, f(s_i)) + \dist_G(s_j, f(s_j)) = 2 + 3 = 5$.)
        It remains to handle all tokens on $s_j \in \mathcal{A}_{12}$, one token at a time, as follows:
        \begin{itemize}
            \item Pick an $s_j \in \mathcal{A}_{12}$ and slide its token one step directly to $s_i$ (which is currently unoccupied). This can be done because $s_i \in K$.
            \item Slide the token currently placed on $u$ to $f(s_j)$. This can be done because $s_i$ now has a token.
            \item Slide the token currently on $s_i$ to $u$. This can be done because $u \in K$.
            \item Remove $s_j$ from $\mathcal{A}_{12}$.
            \item Repeat the process until $\mathcal{A}_{12} = \emptyset$.
        \end{itemize}
        After every token on vertices in $\mathcal{A}_{12}$ is handled, we can simply move the last token on $u$ directly to $f(s_i)$.
        Intuitively, what we are doing in the above process is to slide the token that is originally on $s_i$ to $f(s_j)$ and the token that is originally on $s_j$ to $u$ via $s_i$, all of which can be done using exactly three token-slides along the shortest $s_jf(s_j)$-path $s_js_iuf(s_j)$ (of length exactly $3$)\footnote{Intuitively, one may think of this as in order to achieve $M^\star$ token-slides, here we have three token-slides to use, and as long as we do not use more, we are fine.}.
        Thus, in this case, our constructed $\sfTS$-sequence has length exactly $M^\star$.
        
        \item To complete our construction, we consider the case where $u$ has a token.
        In this case, we simply slide that token to $f(s_i)$ (which can be done because $s_i \in K$ currently has a token) and move to the case where $u$ has no token. 
        After all tokens on vertices in $\mathcal{A}_1$ are \textit{handled}, we simply skip the final step of moving the token on $u$ to $f(s_i)$ as we have already done that from the beginning.
        Again, our constructed $\sfTS$-sequence has length exactly $M^\star$.
    \end{itemize}

\item \textbf{Case~B: $\dist_G(s_i, f(s_i)) = 1$.}
	In this case, $f(s_i)$ may be either in $S$ or in $K$.

	\begin{itemize}
		\item We first consider the case $f(s_i) \in S$.
		In this situation, $\mathcal{A}_2 = \emptyset$, and therefore, we need to \textit{handle} only tokens on vertices in $\mathcal{A}_1$.
		If every vertex in $K$ has a token, and since $|K| \geq 2$, there must be at least one vertex $u \neq s_i \in K$ having a token. 
		Thus, we first directly slide the token on $s_i$ to $f(s_i)$.
		Again, we can partition $\mathcal{A}_1$ into two disjoint subsets $\mathcal{A}_{11}$ and $\mathcal{A}_{12}$, where $\mathcal{A}_{11}$ consists of all $s_j \in \mathcal{A}_1$ such that there exists a shortest $s_jf(s_j)$-path in $G$ that does not contain $u$, and $\mathcal{A}_{12}$ consists of all $s_j \in \mathcal{A}_1$ such that every shortest $s_jf(s_j)$-path in $G$ contains $u$.

		The tokens on vertices $s_j$ in $\mathcal{A}_{11}$ can be \textit{handled} via the $s_jf(s_j)$-shortest paths that do not contain $u$ using Lemma~\ref{lem:alwaysTSsequence}, as at this point $s_i$ has no token, and the \textit{handling process} does not involve $u \in K$. 
		We \textit{handle} the tokens on vertices in $\mathcal{A}_{12}$ similarly as \textbf{Case~A}.
		(At this point, note that $s_j, u$ both have tokens while $s_i$ and $f(s_j)$ do not. We use the same strategy as in \textbf{Case~A} along the $s_jf(s_j)$-shortest path $s_js_iuf(s_j)$.)
		Again, our constructed $\sfTS$-sequence has length exactly $M^\star$ in this case.

		We now consider the case that there exists a vertex $u \in K$ that has no token.
		As before, we partition $\mathcal{A}_1$ into two disjoint subsets $\mathcal{A}_{11}$ and $\mathcal{A}_{12}$ and handle each of them similarly to \textbf{Case~A} except the final step.
		In the final step, as the vertex $u$ is not in any $s_if(s_i)$-shortest path, we do not move the token on $u$ directly to $f(s_i)$ (which may not even be possible if they are not adjacent) but we move it back to $s_i$ (which can be done because $s_i \in K$) and then finally from $s_i$ to $f(s_i)$ (which can be done because $f(s_i)$ is at this point the unique vertex in $D_t$ not having a token).
		As we use two more extra token-slides (sliding a token from $s_i$ to $u$ and later slide the token on $u$ back to $s_i$), our constructed $\sfTS$-sequence has length at most $M^\star + 2$.

		\item We now consider the case $f(s_i) \in K$.
		First, we show that at this point $f(s_i)$ does not have a token.
		Suppose to the contrary that it does, i.e., there exists $\ell$ such that $s_\ell = f(s_i)$.
		In this case, no shortest $s_\ell f(s_\ell)$-path contains $s_i$; otherwise, swapping the targets of $s_i$ and $s_\ell$ would lead to a total sum $\sum_{i=1}^k\dist_G(s_i, f(s_i))$ smaller than $M^\star$.
		Thus, $s_\ell$ must already have been handled before, which means $f(s_i)$ cannot have a token at this point, a contradiction.

		As $f(s_i)$ does not have a token, we directly slide the token on $s_i$ to $f(s_i)$.
		At this point, each $s_j \in \mathcal{A}_1$ (where there exists an $s_jf(s_j)$-shortest path that does not contain $f(s_i)$) can be \textit{handled} using Lemma~\ref{lem:alwaysTSsequence}, as currently there is an $s_jf(s_j)$-shortest path that does not contain $s_i$.
		For the tokens on $s_j \in \mathcal{A}_2$ (where both $s_i$ and $f(s_i)$ are in every $s_jf(s_j)$-path), we can then set $s_i^\prime = f(s_i) = f(s_i^\prime)$ and instead of working with the token-set $s_1, \dots, s_i, \dots, s_k$ and the value $M^\star$, we work with the token-set $s_1, \dots, s_i^\prime, \dots, s_k$ and the value $M^\star - 1$.
		As $s_i^\prime = f(s_i^\prime)$, we can reduce to \textbf{Case~C} below.
	\end{itemize}

	\item \textbf{Case~C: $\dist_G(s_i, f(s_i)) = 0$.}
	In this case, note that $\mathcal{A}_1 = \emptyset$. So we need to only deal with the set $\mathcal{A}_2$. In this case, both $s_i$ and $f(s_i)$ are the same vertex and both belong to $K$.
	If every vertex in $K$ has a token, 
	there must be some $s_j \in \mathcal{A}_2$ such that $f(s_j) \in S$ 
	has no token on it; otherwise, there is no token left to consider. 
	Let $u$ be one of $f(s_j)$'s neighbors in $K$ along a $s_jf(s_j)$-shortest path. 
	(As $G$ is connected, such a vertex $u$ exists.)
	As $u$ has a token, there exists $\ell$ such that $u = s_{\ell}$.
	We then swap the targets of $s_j$ and $s_{\ell}$ and slide the token on $u = s_{\ell}$ to its new target (which was originally $f(s_j)$).
	One can verify by considering all relative positions of $s_j, f(s_j), u = s_{\ell}, f(s_\ell)$ (note that $u$ and $f(s_j)$ are always adjacent in a $s_jf(s_j)$-shortest path) that this ``swapping targets'' procedure does not lead to any total sum $\sum_{i=1}^k\dist_G(s_i, f(s_i))$ smaller than $M^\star$.
	Now, as $u$ has no token, we reduce to the case where not every vertex in $K$ has a token.

	We now consider the case where not every vertex in $K$ has a token.
	In this case, since $|K| \geq 2$, there must be at least one vertex $u \neq s_i (f(s_i))$ in $K$ having no token.
	As before, we partition $\mathcal{A}_2$ into two disjoint subsets $\mathcal{A}_{21}$ and $\mathcal{A}_{22}$ where $\mathcal{A}_{21}$ consists of all $s_j \in \mathcal{A}_2$ such that there exists a shortest $s_jf(s_j)$-path in $G$ which does not contain $u$ (but contains $s_i = f(s_i)$) and $\mathcal{A}_{22}$ consists of all $s_j \in \mathcal{A}_2$ such that every shortest $s_jf(s_j)$-path in $G$ contains $u$ (and $s_i = f(s_i)$).
	We then \textit{handle} $\mathcal{A}_{21}$ (resp., $\mathcal{A}_{22}$) similarly to what we did before for $\mathcal{A}_{11}$ (resp., $\mathcal{A}_{12}$) in \textbf{Case~A} and \textbf{Case~B}.
	Again, since we use two extra token-slides (sliding token from $s_i$ to $u$ and then later sliding token from $u$ back to $s_i$), our constructed $\sfTS$-sequence has length $M^\star + 2$.
	(For example, see the constructed example in the proof of Lemma~\ref{lem:TS-split-shortest}. In that example, $\mathcal{A}_{22} = \emptyset$.)
	
  \end{itemize}
  \end{proof}

It remains to consider the $\sfTJ$ rule.
\begin{lemma}\label{lem:TJ-split-shortest}
	There exists a split graph $G = (K \uplus S, E)$ and two size-$k$ D$2$DSs $D_s, D_t$ of $G$ such that $\displaystyle\opt_{\sfTJ}(G, D_s, D_t) = M^\star_{\sfTJ}(G, D_s, D_t) + 1$, for $k \geq 2$.
\end{lemma}
\begin{proof}

We construct a split graph $G = (K \uplus S, E)$ as follows:

\begin{itemize}
\item The set $K$ contains $3k$ vertices labeled $v_{i1}, v_{i2}$, and $v_{i3}$, for $1 \leq i \leq k$. Vertices in $K$ form a clique in $G$.

\item The set $S$ contains $3k$ vertices labeled $u_{i1}, u_{i2}, w_{i1}$ for $1 \leq i \leq k$. Vertices in $S$ form an independent set in $G$.

\item For each $i \in \{1, \dots, k\}$, we join $u_{i1}$ to both $v_{i1}$ and $v_{i2}$, join $u_{i2}$ to both $v_{i1}$ and $v_{i3}$, and join $w_{i1}$ to $v_{j3}$, where $j = i + 1$ if $i \leq k - 1$ and $j = 1$ if $i = k$.
\end{itemize}

Let $D_s = \bigcup_{i=1}^k \{u_{i1}\}$ and $D_t = \bigcup_{i=1}^k \{w_{i1}\}$.
	
		\begin{figure}[ht]
			\centering
			\includegraphics[width=0.8\textwidth]{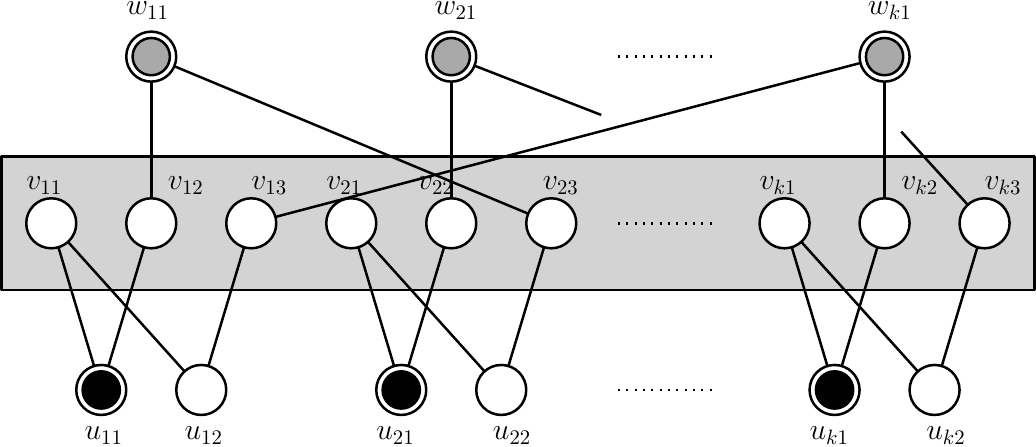}
			\caption{Construction of a split graph $G = (K \uplus S, E)$ satisfying Lemma~\ref{lem:TJ-split-shortest}. Vertices in the light gray box are in $K$. Tokens in $D_s$ and $D_t$ are respectively marked by black and gray circles.}
			\label{fig:TJ-split}
		\end{figure}

  To see that $D_s$ is a D$2$DS of $G$, note that each $u_{i1}$ 2-dominates every vertex in $K \cup \{u_{i1}\} \cup \{u_{i2}\} \cup \{w_{i1}\}$. To see that $D_t$ is a D$2$DS of $G$, note that each $w_{i1}$ 2-dominates every vertex in $K \cup \{u_{i1}\} \cup \{u_{j2}\} \cup \{w_{i1}\}$, where $j = i + 1$ if $i \leq k - 1$ and $j = 1$ if $i = k$.

Let $M^\star = M^\star_{\sfTJ}(G, D_s, D_t)$. Observe that any $\sfTJ$-sequence of length exactly $M^\star$ must begin with a direct token-jump from some $u_{i1}$ to some $w_{j1}$ for some $1 \leq i, j \leq k$. We claim that such a move cannot be performed:

\begin{itemize}
	\item If the token on $u_{i1}$ is moved directly to any $w_{j1}$ where $i \neq j$, the vertex $u_{i1}$ is not $2$-dominated by any token in the resulting token-set.
	\item If the token on $u_{i1}$ is moved directly to $w_{i1}$, the vertex $u_{i2}$ is not $2$-dominated by any token in the resulting token-set.
\end{itemize}

On the other hand, a $\sfTJ$-sequence of length exactly $M^\star + 1$ can be constructed: first jumping the token on $u_{11}$ to $v_{11}$, then for $2 \leq i \leq k$, directly jumping the token on $u_{i1}$ to $w_{i1}$, and finally jumping the token on $v_{11}$ to $w_{11}$. As before, since a token is always placed at $v_{11} \in K$ after the first token-jump and before the final one, the above sequence of token-jumps is indeed a $\sfTJ$-sequence in $G$.
\end{proof}

\begin{lemma}\label{lem:TJ-split-shortest-bound}
	For any instance $(G, D_s, D_t)$ of \textsc{D$2$DSR} under $\sfTJ$ on split graphs, $\opt_{\sfTJ}(G, D_s, D_t) \leq M^\star_{\sfTJ}(G, D_s, D_t) + 1$.
\end{lemma}
\begin{proof}
We use the same strategy as in the proof of Lemma~\ref{lem:TJ-split-shortest}. Let $M^\star = M^\star_{\sfTJ}(G, D_s, D_t)$. To construct a $\sfTJ$-sequence of length exactly $M^\star + 1$ between any pair of $D_s$ and $D_t$ of $G = (K \uplus S, E)$, we proceed as follows:

\begin{enumerate}[(1)]
	\item  If $D_s \cap K = \emptyset$, we first move a token $t$ in $D_s$ to some vertex in $K$ (which takes one token-jump).
	\item We then move every other token in $D_s \setminus D_t$ to some vertex in $D_t \setminus D_s$ (which takes $M^\star - 1$ token-jumps).
	\item Finally, we move $t$ to the unique remaining unoccupied vertex in $D_t$ (which takes one token-jump).
\end{enumerate}

Since $t$ is always located in $K$ except possibly before the initial jump and after the final jump, our constructed sequence is a $\sfTJ$-sequence, and it has length $M^\star + 1$.
\qed\end{proof}

\subsection{Trees}
\label{sec:trees}

By Corollary~\ref{cor:dually-chordal}, it follows that \textsc{D$r$DSR} under $\sfTJ$ on trees can be solved in quadratic time. 
In this section, we improve the running time as follows.

\begin{theorem}\label{thm:TJ-trees}
	\textsc{D$r$DSR} under $\sfTJ$ on trees can be solved in linear time for any $r \geq 2$.
\end{theorem}

To prove this theorem, we extend the idea of Haddadan et al. \cite{HaddadanIMNOST16} for $r = 1$ under $\sfTAR$ and the linear-time algorithm of Kundu and Majumder~\cite{KunduM16} for finding a minimum D$r$DS on trees. 
In particular, we employ a simpler implementation of Kundu and Majumder's algorithm presented by Abu-Affash, Carmi, and Krasin~\cite{Abu-AffashCK22}. 
Specifically, based on the minimum D$r$DS $D^\star$ obtained from the implementation of Abu-Affash, Carmi, and Krasin, we construct a partition $\mathbb{P}(T)$ of $T$ consisting of $\gamma_r(T)$ vertex-disjoint subtrees, each containing exactly one vertex of $D^\star$. 
(Haddadan~et~al. called such a set $D^\star$ a \textit{canonical} dominating set.) 
For convenience, we denote by $C_x$ the member of $\mathbb{P}(T)$ whose intersection with $D^\star$ is the vertex $x$. 
We claim that $\mathbb{P}(T)$ satisfies the following property: for any D$r$DS $D$ of $G$, each member of $\mathbb{P}(T)$ contains at least one vertex in $D$. 
Using this property, one can design a linear-time algorithm for constructing a $\sfTJ$-sequence between any pair of size-$k$ D$r$DSs $D_s, D_t$ of $G$. 
The key idea is to transform both $D_s$ and $D_t$ into some D$r$DS $D$ that contains $D^\star$. 
For instance, to transform $D_s$ into $D$, for each subtree $C_x \in \mathbb{P}(T)$ for $x \in D^\star$, we move any token in $D_s \cap V(C_x)$ to $x$. 
If we handle each subtree $C_x$ based on the order of subtrees added to $\mathbb{P}(T)$ in our modified implementation, such a transformation will form a $\sfTJ$-sequence in $T$. 
After this procedure, we obtain a set of tokens $D^\prime$ that contains $D^\star$, and since $D^\star$ is a minimum D$r$DS of $G$, transforming $D^\prime$ into $D$ under $\sfTJ$ can now be done easily: until there are no tokens to move, repeatedly take a token in $D^\prime - D$, move it to some vertex in $D - D^\prime$, and update both $D$ and $D^\prime$.

We now define some notations and, for the sake of completeness, describe the algorithm of Abu-Affash, Carmi, and Krasin~\cite{Abu-AffashCK22}.
In a graph $G$, for a vertex subset $D \subseteq V(G)$ and a vertex $u \in V(G)$, we define $\delta_D(u) = \min_{v \in D}\dist_G(u, v)$ and call it the \textit{distance} between $u$ and $D$.
Observe that a vertex $u$ is $r$-dominated by $D$ if $\delta_D(u) \leq r$ and therefore $D$ is a D$r$DS of $G$ if for every $u \in V(G)$ we have $\delta_D(u) \leq r$.
For a $n$-vertex tree $T$, let $T_u$ be the \textit{rooted form} of $T$ when regarding the vertex $u \in V(T)$ as the root.
For each $v \in V(T_u)$, we denote by $T_v$ the subtree of $T_u$ rooted at $v$.
In other words, $T_v$ is the subtree of $T_u$ induced by $v$ and its descendants.
We also define $h(T_v) = \max_{w \in V(T_v)}\dist_{T_u}(v, w)$ and call it the \textit{height} of $T_v$.
In other words, $h(T_v)$ is the largest distance from $v$ to a vertex in $T_v$.
The set of children of $v$ in $T_u$ is denoted by $\textit{child}(v)$.

The algorithm is described in Algorithm~\ref{algo:minDrDS}.
In short, in each iteration, it finds a subtree $T_v$ of height exactly $r$, adds $v$ to $D^\star$, and removes all the vertices of $T_u$ that are in $N^r_{T_u}[v]$.
To implement the algorithm in $O(n)$ time, a modified version of the depth-first search (DFS) algorithm was used in~\cite{Abu-AffashCK22}
(Function $\mathtt{ModifiedDFS}$ in Algorithm~\ref{algo:minDrDS}).
The procedure $\mathtt{ModifiedDFS}$ visits the vertices of $T_u$ starting from the root $u$ and recursively visits each of its children, which means vertices in $D^\star$ would be added in a ``bottom-up'' fashion.
In each recursive call $\mathtt{ModifiedDFS}(v)$, if $h(T_v) = r$ then $v$ is added to $D^\star$ and since all vertices in $T_v$ is $r$-dominated by $v$, we remove them from $T_u$ and return $\delta_{D^\star}(v) = 0$.
Otherwise ($h(T_v) \neq r$), we call $\mathtt{ModifiedDFS}(w)$ for each child $w$ of $v$, and we update $\delta_{D^\star}(v)$ and $h(T_v)$ according to these calls.
When these calls return, we have $h(T_v) \leq r$.
Then we check whether $\delta_{D^\star}(v) + h(T_v) \leq r$.
If so (which means the current $D^\star$ $r$-dominates $v$ and all its descendants in the original rooted tree $T_u$), we remove all the vertices of $T_v$ from $T_r$ and return $\delta_{D^\star}(v)$.
Otherwise ($\delta_{D^\star}(v) + h(T_v) > r$), we check again whether $h(T_v) = r$ (in case the descendants reduced the height of $T_v$ to $r$).
If so, we add $v$ to $D^\star$, remove all the vertices of $T_v$ from $T_u$, and return $\delta_{D^\star}(v) = 0$.
Otherwise ($h(T_v) < r$), we return $\infty$.
Finally, when $\delta_{D^\star}(u) = \infty$, we add $u$ to $D^\star$.

\begin{algorithm}[ht]
	\KwIn{A tree $T_u$ rooted at $u$.}
	\KwOut{A minimum distance-$r$ dominating set $D^\star$ of $T_u$.}
	\SetArgSty{textbb} 
	\DontPrintSemicolon
	
	$D^\star \gets \emptyset$\;
	\For{each $v \in V(T_u)$}{
		compute $h(T_v)$\;
		$\delta_{D^\star}(v) \gets \infty$\;
	}
	$\delta_{D^\star}(u) \gets \mathtt{ModifiedDFS}(u)$\;
	\If{$\delta_{D^\star}(u) = \infty$}{
		$D^\star \gets D^\star + u$\;
	}
	\Return{$D^\star$}
	
	\BlankLine
	
	\SetKwFunction{MDFS}{ModifiedDFS}
	\SetKwProg{Fn}{Function}{:}{}
	
	\Fn{\MDFS{$v$}}{
		\If{$h(T_v) = r$}{
			$D^\star \gets D^\star + v$\;
			$T_u \gets T_u - T_v$\;
			$h(T_v) \gets -1$\;
			$\delta_{D^\star}(v) \gets 0$\;
		}
		\Else(\tcp*[f]{$h(T_v) > r$}){
			\For{each $w \in \textit{child}(v)$}{
				\If{$h(T_w) \geq r$}{
					$\delta_{D^\star}(v) \gets \min\{\delta_{D^\star}(v), \mathtt{ModifiedDFS}(w) + 1\}$
				}
			}
			$h(T_v) \gets \max\{h(T_w) + 1: w \in \textit{child}(v)\}$\tcp*{updating $h(T_v)$}
			\If{$h(T_v) + \delta_{D^\star}(v) \leq r$}{
				$T_u \gets T_u - T_v$\;
				$h(T_v) \gets -1$\;
			}
			\If{$h(T_v) = r$}{
				$D^\star \gets D^\star + v$\;
				$T_u \gets T_u - T_v$\;
				$h(T_v) \gets -1$\;
				$\delta_{D^\star}(v) \gets 0$\;
			}
			\Else{
				$\delta_{D^\star}(v) \gets \infty$\;
			}
		}
		
		\Return{$\delta_{D^\star}(v)$}
	}
	
	\caption{$\mathtt{MinDrDSTree}(T_u)$}
	\label{algo:minDrDS}
\end{algorithm}

To illustrate Algorithm~\ref{algo:minDrDS}, we consider the example from~\cite{Abu-AffashCK22} for $r = 2$ with the tree $T_u$ rooted at $u = 1$ as described in \figurename~\ref{fig:TJ-trees}.
The first vertex added to $D^\star$ is $7$ in $\mathtt{ModifiedDFS}(7)$, since $h(T_7) = 2$.
In this call, we remove $T_7$ from $T_u$, update $h(T_7) = -1$ and return $\delta_{D^\star}(7) = 0$ to $\mathtt{ModifiedDFS}(4)$.
In $\mathtt{ModifiedDFS}(4)$, we update $\delta_{D^\star}(4) = 1$ and, after traversing vertices $6$ and $8$, $h(T_4) = 1$, and since $h(T_4) + \delta_{D^\star}(4) = 2 = r$ and $7$ is the latest vertex added to $D^\star$, we remove $T_4$ from $T_u$ and return $\delta_{D^\star}(4) = 1$ to $\mathtt{ModifiedDFS}(2)$.
In $\mathtt{ModifiedDFS}(2)$, since $h(T_2) = 3 > r$, we call $\mathtt{ModifiedDFS}(5)$ which adds $5$ to $D^\star$, removes $T_5$ from $T_u$, and returns $\delta_{D^\star}(5) = 0$.
Then, we update $\delta_{D^\star}(2) = 1$ and $h(T_2) = 0$, and since $h(T_2) + \delta_{D^\star}(2) = 1 < r$ and $5$ is the latest vertex added to $D^\star$, we remove $T_2$ from $T_u$, and return $\delta_{D^\star}(2) = 1$ to $\mathtt{ModifiedDFS}(1)$.
In $\mathtt{ModifiedDFS}(1)$, since $\delta_{D^\star}(1) = 2$ and, after traversing $3$, $h(T_1) = 1$, we return $\delta_{D^\star}(1) = \infty$ to Algorithm~\ref{algo:minDrDS}, and therefore, we add $1$ to $D^\star$.

\begin{figure}[ht]
	\centering
	\includegraphics[width=0.35\textwidth]{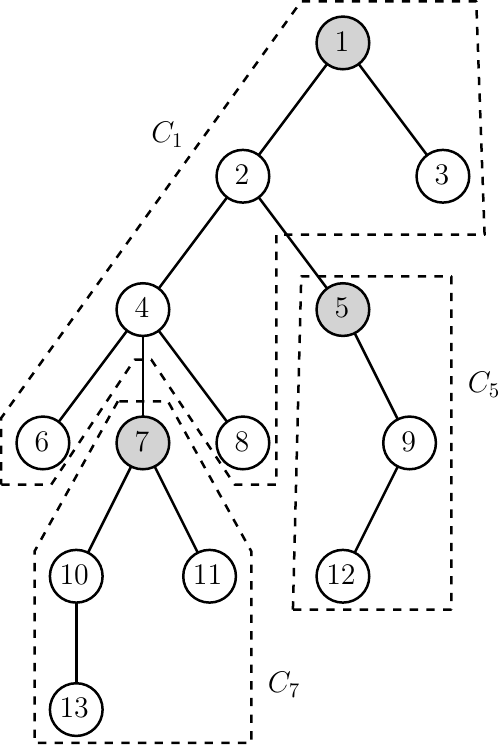}
	\caption{A tree $T_u$ rooted at $u = 1$. For $r = 2$, Algorithm~\ref{algo:minDrDS} returns $D^\star = \{7, 5, 1\}$. A partition $\mathbb{P}(T_u) = \{C_7, C_5, C_1\}$ of $T_u$ is also constructed.}
	\label{fig:TJ-trees}
\end{figure}

We now describe how to construct our desired partition $\mathbb{P}(T_u)$.
Recall that $\mathbb{P}(T_u)$ is nothing but a collection of vertex-disjoint subtrees whose union is the original tree $T_u$.
Suppose that $D^\star$ is the minimum D$r$DS of $T_u$ obtained from Algorithm~\ref{algo:minDrDS} and furthermore assume that vertices of $D^\star$ are ordered by the time they were added to $D^\star$.
For each $v \in D^\star$, we define $C_v$ (the unique member of $\mathbb{P}(T_u)$ containing $v$) as $T_v$ (the subtree of $T_u$ rooted at $v$) and then delete $T_v$ from $T_u$.
\figurename~\ref{fig:TJ-trees} illustrates how to construct $\mathbb{P}(T_u)$ in the above example.
From the construction, it is clear that each member of $\mathbb{P}(T_u)$ contains exactly one vertex from $D^\star$.
We say that two subtrees $C_x, C_y$ in $\mathbb{P}(T_u)$ are \textit{adjacent} if there exists $v \in V(C_x)$ and $w \in V(C_y)$ such that $vw \in E(T_u)$.
If a subtree contains the root $u$ then we call it the \textit{root subtree}.
Otherwise, if a subtree has exactly one adjacent subtree then we call it a \textit{leaf subtree} and otherwise an \textit{internal subtree}.

We now claim that the constructed partition $\mathbb{P}(T_u)$ satisfies the following property.
\begin{lemma}\label{lem:partition-tree}
	Let $D$ be any D$r$DS of $T_u$.
	Then, $D \cap V(C_v) \neq \emptyset$ holds for every $v \in D^\star$. 
\end{lemma}
\begin{proof}
	We claim that for each $v \in D^\star$, one can find a vertex $v^\prime \in V(C_v)$ such that $N^r_T[v^\prime] \subseteq V(C_v)$.
	For each $v \in D^\star$, let $D^\star_v$ be the set of all vertices added to $D^\star$ before $v$.
	
	If $C_v$ is a leaf subtree, we take any leaf in $C_v$ of distance exactly $r$ from $v$ and regard it as $v^\prime$.
	Clearly, $v^\prime$ is also a leaf of $T_u$ and is not $r$-dominated by any vertex outside $C_v$, i.e., $N^r_T[v^\prime] \subseteq V(C_v)$.
	
	If $C_v$ is an internal subtree, we describe how to find our desired $v^\prime$.
	From Algorithm~\ref{algo:minDrDS}, since $v$ is the next vertex added to $D^\star_v$, it follows that there must be some vertex in $V(C_v)$ not $r$-dominated by any member of $D^\star_v$; we take $v^\prime$ to be the one having maximum distance from $v$ among all those vertices.
	By definition, $v^\prime$ is clearly not $r$-dominated by any vertex in a $C_w$ where $w \in D^\star_v$.
	Since $C_v$ is an internal subtree, by Algorithm~\ref{algo:minDrDS}, the distance between $v$ and $v^\prime$ must be exactly $r$ and therefore no vertex in a $C_w$, where $w \in D^\star - D^\star_v - v$, $r$-dominates $v$.
	(Recall that by Algorithm~\ref{algo:minDrDS}, since $v$ is added to $D^\star_v$, the current subtree $T_v$ must have height exactly $r$.)
	Thus, $N_T^r[v^\prime] \subseteq V(C_v)$.
	
	If $C_v$ is the root subtree, again we can choose $v^\prime$ using exactly the same strategy as in the case for internal subtrees.
	The main difference here is that, by Algorithm~\ref{algo:minDrDS}, the distance between $v^\prime$ and $v$ may not be exactly $r$.
	However, since $C_v$ contains the root $u$, $v$ is the last vertex added to $D^\star$.
	(Intuitively, this means $C_v$ has no ``parent subtree'' above it.)
	Therefore, in order to show that $N_T^r[v^\prime] \subseteq V(C_v)$, it suffices to show that no vertex in a $C_w$, where $w \in D^\star_v$, $r$-dominates $v^\prime$.
	Indeed, this clearly holds by definition of $v^\prime$.
	
	We are now ready to prove the lemma.
	Suppose to the contrary that there exists $v \in D^\star$ such that $D \cap V(C_v) = \emptyset$.
	Then, $D$ does not $r$-dominate $v^\prime$---a vertex in $C_v$ with $N_T^r[v^\prime] \subseteq V(C_v)$.
	This contradicts the assumption that $D$ is a D$r$DS.
	Our proof is complete.
\end{proof}

The following lemma is crucial in proving Theorem~\ref{thm:TJ-trees}.
\begin{lemma}\label{lem:TJ-trees}
	Let $D$ be an arbitrary D$r$DS of $T_u$.
	Let $D^\prime$ be any D$r$DS of $T_u$ that contains $D^\star$, i.e., $D^\star \subseteq D^\prime$.
	Then, in $O(n)$ time, one can construct a $\sfTJ$-sequence $\calS$ in $T_u$ between $D$ and $D^\prime$.
\end{lemma}
\begin{proof}
	We construct $\calS$ as follows. 
	Initially, $\calS = \emptyset$.
	\begin{enumerate}[{\bf {Step} 1:}]
		\item For each $v \in D^\star$, let $x$ be any vertex in $D \cap V(C_v)$.
		From Lemma~\ref{lem:partition-tree}, such a vertex $x$ exists.
		We append $x \reconf[\sfTJ]{T_u} v$ to $\calS$ and assign $D \gets D - x + v$.
		(After this step, clearly $D^\star \subseteq D \cap D^\prime$.)
		
		\item Let $x \in D - D^\prime$ and $y \in D^\prime - D$. 
		We append $x \reconf[\sfTJ]{T_u} y$ to $\calS$ and assign $D \gets D - x + y$.
		Repeat this step until $D = D^\prime$.
	\end{enumerate}
	
	For each $v \in D^\star$, let $D^\star_v$ be the set of all vertices added to $D^\star$ before $v$.
	Since any vertex $r$-dominated by $x$ and not in $C_v$ is $r$-dominated by either $v$ or a member of $D^\star_v$, any move performed in \textbf{Step~1} results a new D$r$DS of $T_u$.
	Note that after {\bf Step~1}, $D^\star \subseteq D \cap D^\prime$.
	Thus, any move performed in {\bf Step~2} results a new D$r$DS of $T_u$.
	In short, $\calS$ is indeed a $\sfTJ$-sequence in $T_u$.
	In the above construction, as we ``touch'' each vertex in $D$ at most once, the running time is indeed $O(n)$.
\end{proof}

Using Lemma~\ref{lem:TJ-trees}, it is not hard to prove Theorem~\ref{thm:TJ-trees}.
More precisely, let $(T, D_s, D_t)$ be an instance of \textsc{D$r$DSR} under $\sfTJ$ where $D_s$ and $D_t$ are two D$r$DSs of a tree $T$.
By Lemma~\ref{lem:TJ-trees}, one can immediately decide if $(T, D_s, D_t)$ is a yes-instance by comparing the sizes of $D_s$ and $D_t$: if they are of the same size then the answer is ``yes'' and otherwise it is ``no''.
Moreover, in a yes-instance, Lemma~\ref{lem:TJ-trees} allows us to construct in linear time a $\sfTJ$-sequence (which is not necessarily a shortest one) between $D_s$ and $D_t$.

\section{Hardness Results}
\label{sec:hardness}

We begin this section with the following remark.
It is well-known that for any computational problem in $\ttNP$, any of its reconfiguration variants is in $\ttPSPACE$ (see the proof of \cite[Theorem~1]{ItoDHPSUU11}). 
Hence, as \textsc{D$r$DS} is in  $\ttNP$~\cite{ChangN84}, \textsc{D$r$DSR} is in $\ttPSPACE$.
As a result, when proving the $\ttPSPACE$-completeness of \textsc{D$r$DSR} on a certain graph class, it suffices to show a polynomial-time reduction.

\subsection{Base Problems for Hardness Reduction}
\label{sec:pspacec-base-probs}

\subsubsection*{(Minimum) Vertex Cover Reconfiguration.}
Recall that a \textit{vertex cover} of a graph $G$ is a vertex subset $C$ such that for every edge $e$ of $G$, at least one endpoint of $e$ is in $C$.
An \textit{independent set} of a graph $G$ is a vertex subset $I$ such that no two vertices in $I$ are joined by an edge.
We use $\tau(G)$ and $\alpha(G)$ to respectively denote the size of a minimum vertex cover and a maximum independent set of $G$.
In a \textsc{Vertex Cover Reconfiguration (VCR)}'s instance $(G, C_s, C_t)$ under $\sfR \in \{\sfTS, \sfTJ, \sfTAR\}$, two vertex covers $C_s$ and $C_t$ a graph $G$ are given, and the question is to decide whether there is a $\sfR$-sequence between $C_s$ and $C_t$.
The \textsc{Independent Set Reconfiguration (ISR)} problem under $\sfR \in \{\sfTS, \sfTJ, \sfTAR\}$ is defined similarly.
In the \textsc{Minimum Vertex Cover Reconfiguration (Min-VCR)} (resp., \textsc{Maximum Independent Set Reconfiguration (Max-ISR)}) problem, all inputs are minimum vertex covers (resp., maximum independent sets). 
Since $I$ is an independent set of $G$ if and only if $V(G) - I$ is a vertex cover, from the classic complexity viewpoint, \textsc{VCR} (resp., \textsc{Min-VCR}) and \textsc{ISR} (resp., \textsc{Max-ISR}) share the same complexity status under the same reconfiguration rule.
In this paper, we are interested in using some variants of \textsc{Min-VCR} as the bases for our hardness reduction.

Ito et al. \cite{ItoDHPSUU11} proved that \textsc{ISR} is $\ttPSPACE$-complete on general graphs under $\sfTAR$.
Since $\sfTJ$ and $\sfTAR$ are somewhat equivalent~\cite[Theorem~1]{KaminskiMM12}---any reconfiguration sequence under $\sfTJ$ between two size-$k$ independent sets can be converted to a reconfiguration sequence under $\sfTAR$ whose members are independent sets of size at least $k - 1$ and vice versa---the result of Ito et al. can be applied for $\sfTJ$.
Indeed, Ito et al. constructed reconfiguration sequences under $\sfTAR$ which contain only independent sets of maximum size $k$ and those of size $k-1$. 
When converting to reconfiguration sequences under $\sfTJ$ as described in~\cite[Theorem~1]{KaminskiMM12}, all independent sets have \textit{maximum} size.
Additionally, note that any $\sfTJ$-move applied to a maximum independent set is indeed also a $\sfTS$-move\footnote{Another way to see this is as follows. From Ito et al.'s construction, observe that under $\sfTJ$ no token can ever leaves its corresponding gadget, which is either an edge or a triangle, and therefore any $\sfTJ$-move is indeed also a $\sfTS$-move.}, because any maximum independent set is also a dominating set.
Thus, the result of Ito et al. can be applied for $\sfTS$ even when all independent sets are of maximum size.
In short, \textsc{Max-ISR} is $\ttPSPACE$-complete on general graphs under $\sfR \in \{\sfTS, \sfTJ, \sfTAR\}$ and therefore so is \textsc{Min-VCR}.

Hearn and Demaine~\cite{HearnD05} proved that \textsc{ISR} is $\ttPSPACE$-complete on planar graphs of maximum degree three and later van der Zanden~\cite{Zanden15} extended their results for those graphs with the additional ``bounded bandwidth'' restriction. 
Hearn and Demain expilicitly considered only the $\sfTS$ rule and later Bonsma et al.~\cite{BonsmaKW14} observed that their result also holds under $\sfTJ$ and $\sfTAR$ (using a similar argument as what we mentioned for general graphs).
Again, since their proofs only involve \textit{maximum} independent sets, \textsc{Max-ISR} is $\ttPSPACE$-complete on planar graphs of maximum degree three and bounded bandwidth under $\sfR \in \{\sfTS, \sfTJ, \sfTAR\}$ and therefore so is \textsc{Min-VCR}.

\subsubsection*{Nondeterministic Constraint Logic (NCL).}
In an \textit{NCL constraint graph} (or \textit{NCL graph} for short), first introduced by Hearn and Demaine~\cite{HearnD05}, each edge has weight either $1$ (red, thick) or $2$ (blue, thin).
A \textit{configuration} of an NCL graph is an orientation of its edges satisfying that the \textit{in-weight} (i.e., the sum of all weights of edges that are directed inward) at any vertex is at least two.
Two NCL configurations of an NCL graph are \textit{adjacent} if one can be obtained from the other by reversing the direction of exactly one edge.
Given an NCL graph and two configurations $C_s, C_t$, the NCL \textit{configuration-to-configuration (C2C)} variant asks whether there is a sequence of adjacent configurations between $C_s$ and $C_t$, i.e., a sequence $\langle C_s = C_1, C_2, \dots, C_p = C_t \rangle$ where all $C_i$ ($1 \leq i \leq p$) are NCL configurations and $C_i$ and $C_{i+1}$ ($1 \leq i \leq p-1$) are adjacent.

An \textit{NCL \textsc{And/Or} constraint graph} is an NCL graph where only two types of vertices, namely \textsc{And} vertices and \textsc{Or} vertices, are allowed. 
(See \figurename~\ref{fig:NCL-def}.)

\begin{itemize}
	\item In an NCL \textsc{And} vertex $u$ (\figurename~\ref{fig:NCL-def}(b)), there are three incident edges: two of which have weight $1$ (red) and one has weight $2$ (blue).
	To direct the weight-$2$ (blue) edge outward, \textit{both} two weight-$1$ (red) edges must be directed inward.
	It is not required that when two weight-$1$ (red) edges are directed inward, the weight-$2$ (blue) edge must be directed outward.
 This mimics the behavior of a logical \textsc{And} gate.	
	\item In an NCL \textsc{Or} vertex $v$ (\figurename~\ref{fig:NCL-def}(c)), there are three incident edges, each of which have weight $2$ (blue).
	To direct any weight-$2$ (blue) edge outward, \textit{at least one} of the other two weight-$2$ (blue) edges must be directed inward.
	This mimics the behavior of a logical \textsc{Or} gate.
\end{itemize}

\begin{figure}[ht]
	\centering
	\includegraphics[width=0.7\textwidth]{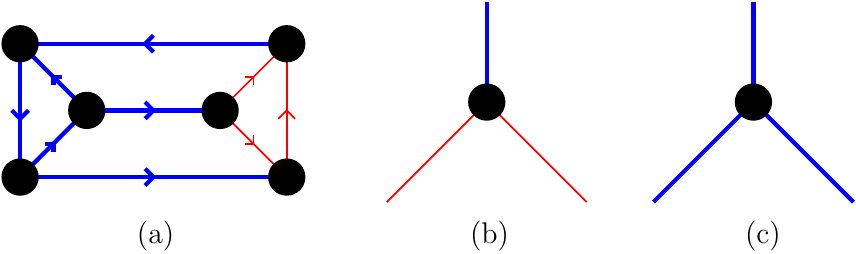}
	\caption{(a) A configuration of an NCL \textsc{And/Or} constraint graph, (b) NCL \textsc{And} vertex, (c) NCL \textsc{Or} vertex.}
	\label{fig:NCL-def}
\end{figure}

For simplicity, we call the NCL C2C variant as the \textit{NCL problem} and an NCL \textsc{And/Or} constraint graph as an \textit{NCL graph}.
It is well-known that the NCL problem remains $\ttPSPACE$-complete even on NCL graphs that are both planar and bounded bandwidth~\cite{Zanden15}.

\subsection{Planar Graphs}
\label{sec:planar}

In this section, we consider planar graphs.
We remark that one can extend the known results for \textsc{DSR} \cite{BonamyDO21,HaddadanIMNOST16} on planar graphs of maximum degree six and bounded bandwidth to show the $\ttPSPACE$-completeness of D$r$DSR ($r \geq 2$) under both $\sfTS$ and $\sfTJ$ on the same graph class.

\begin{theorem}\label{thm:planar-maxdeg6}
        \textsc{D$r$DSR} under $\sfR \in \{\sfTS, \sfTJ\}$ on planar graphs of maximum degree six and bounded bandwidth is $\ttPSPACE$-complete for any $r \geq 2$.
\end{theorem}

\begin{proof}
	We give a polynomial-time reduction from \textsc{Min-VCR} on planar graphs of maximum degree three and bounded bandwidth.
	Our reduction extends the classic reduction from \textsc{Vertex Cover} to \textsc{Dominating Set}~\cite{GareyJohson1979}.
	This reduction has also been modified for showing the hardness of the problem for $r = 1$ (i.e., \textsc{Dominating Set Reconfiguration}) by Haddadan et al. \cite{HaddadanIMNOST16} under $\sfTAR$ and later by Bonamy et al. \cite{BonamyDO21} under $\sfTS$.
	Let $(G, C_s, C_t)$ be an instance of \textsc{Min-VCR} under $\sfR$ where $C_s, C_t$ are two minimum vertex covers of a planar graph $G$ of maximum degree three and bounded bandwidth.
	We will construct an instance $(G^\prime, D_s, D_t)$ of \textsc{D$r$DSR} under $\sfR$ where $D_s$ and $D_t$ are two D$r$DSs of a planar graph $G^\prime$ of maximum degree six and bounded bandwidth.
	
	Suppose that $V(G) = \{v_1, \dots, v_n\}$.
	We construct $G^\prime$ from $G$ as follows.
	For each edge $v_iv_j \in E(G)$, add a new path $P_{ij} = x_{ij}^0x_{ij}^1\dots, x_{ij}^{2r}$ of length $2r$ ($1 \leq i, j \leq n$) with $x_{ij}^0 = v_i$ and $x_{ij}^{2r} = v_j$.
	Observe that $x_{ij}^p = x_{ji}^{2r-p}$ for $0 \leq p \leq 2r$.
	Intuitively, $G^\prime$ is obtained from $G$ by replacing each edge of $G$ by a cycle $\mathcal{C}_{ij}$ of length $2r+1$ formed by the path $P_{ij}$ and the edge $uv = v_iv_j$.
	We define $D_s = C_s$ and $D_t = C_t$.
	Clearly, this construction can be done in polynomial time.
	(See \figurename~\ref{fig:planar}.)
	
	\begin{figure}[ht]
		\centering
		\includegraphics[width=0.7\textwidth]{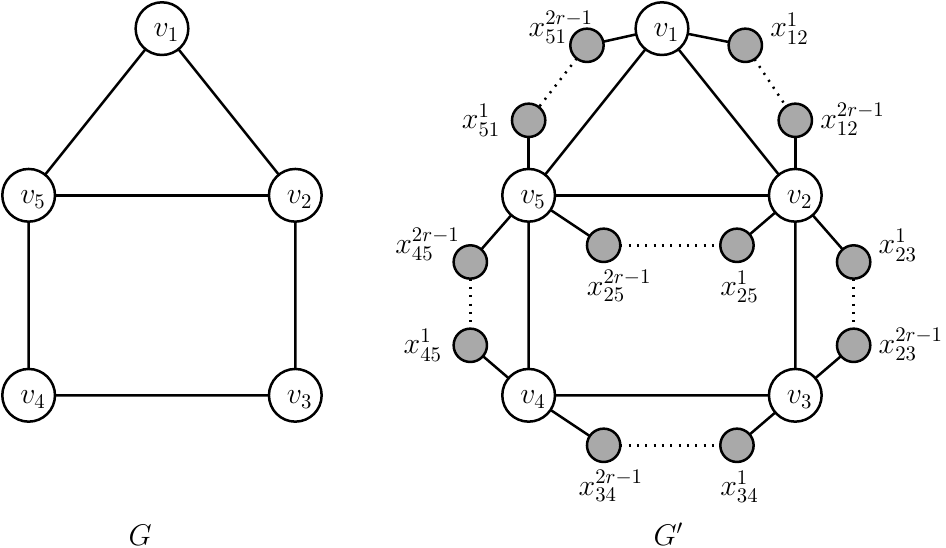}
		\caption{An example of constructing $G^\prime$ from a planar, subcubic, and bounded bandwidth graph $G$ in the proof of Theorem~\ref{thm:planar-maxdeg6}. Vertices in $V(G^\prime) - V(G)$ are marked with the gray color. Each dotted path is of length $2r-2$.}
		\label{fig:planar}
	\end{figure}
	
	We now in the next two lemmas that $(G^\prime, D_s, D_t)$ is indeed our desired \textsc{D$r$DSR}'s instance.
	\begin{claim}
		$G^\prime$ is planar, of maximum degree six, and bounded bandwidth.
	\end{claim}
	\begin{proof}
		It follows directly from the construction that $G^\prime$ is planar and has maximum degree six.
		Since the number of edges of $G^\prime$ is exactly $(2r+1)$ times the number of edges of $G$, the bandwidth of $G^\prime$ increases (comparing to that of $G$) only by a constant multiplicative factor, which implies that $G^\prime$ is a bounded bandwidth graph.
	\end{proof}

	\begin{claim}
		Any minimum vertex cover of $G$ is also a minimum D$r$DS of $G^\prime$.
		Consequently, both $D_s$ and $D_t$ are minimum D$r$DSs of $G^\prime$.
	\end{claim}
	\begin{proof}
		It follows directly from the construction that both $D_s$ and $D_t$ are D$r$DSs of $G^\prime$ (and so is any minimum vertex cover of $G$).
		To show that all minimum vertex covers of $G$ are also minimum D$r$DSs of $G^\prime$, we prove that $\tau(G) = \gamma_r(G^\prime)$ where $\tau(G)$ and $\gamma_r(G^\prime)$ are respectively the size of a minimum vertex cover of $G$ and a minimum D$r$DS of $G^\prime$.
		Since any minimum vertex cover of $G$ is also a D$r$DS of $G^\prime$, we have $\tau(G) \geq \gamma_r(G^\prime)$.
		On the other hand, from the construction of $G^\prime$, observe that for any pair $i, j \in \{1, \dots, n\}$ with $v_iv_j \in E(G)$, the vertex $x_{ij}^r$ (whose distance from both $v_i$ and $v_j$ is exactly $r$) can only be $r$-dominated by some vertex in $V(\mathcal{C}_{ij})$, which implies that one needs at least $\tau(G)$ tokens to $r$-dominate $V(G^\prime)$.
		Therefore, $\gamma_r(G^\prime) \geq \tau(G)$.
	\end{proof}
	
	We are now ready to show the correctness of our reduction. (Claims~\ref{clm:TS-planar} and~\ref{clm:TJ-planar}.)
	\begin{claim}\label{clm:TS-planar}
		Under $\sfTS$, $(G, C_s, C_t)$  is a yes-instance if and only if $(G^\prime, D_s, D_t)$ is a yes-instance.
	\end{claim}
	\begin{proof}
		\begin{itemize}
			\item[($\Rightarrow$)] Let $\calS$ be a $\sfTS$-sequence in $G$ between $C_s$ and $C_t$.
			Since any minimum vertex cover of $G$ is also a minimum D$r$DS of $G^\prime$, the sequence $\calS^\prime$ obtained by replacing each move $u \reconf[\sfTS]{G} v$ in $\calS$ by $u \reconf[\sfTS]{G^\prime} v$ is also a $\sfTS$-sequence in $G^\prime$ between $D_s = C_s$ and $D_t = C_t$.
			
			\item[($\Leftarrow$)] Let $\calS^\prime$ be a $\sfTS$-sequence in $G^\prime$ between $D_s$ and $D_t$.
			We construct a sequence of token-slides $\calS$ in $G$ between $C_s = D_s$ and $C_t = D_t$ as follows.
			Initially, $\calS = \emptyset$.
			For each move $u \reconf[\sfTS]{G^\prime} v$ in $\calS^\prime$, we consider the following cases.
			\begin{enumerate}[{\bf {Case} 1:}]
				\item {\bf $u \in V(G)$ and $v \in V(G)$.} It must happen that $u = v_i$ and $v = v_j$ for some $i, j \in \{1, \dots, n\}$ such that $v_iv_j \in E(G)$. We append $v_i \reconf[\sfTS]{G} v_j$ to $\calS$.
				
				\item {\bf $u \in V(G)$ and $v \in V(G^\prime) - V(G)$.} Do nothing.
				
				\item {\bf $u \in V(G^\prime) - V(G)$ and $v \in V(G)$.} It must happen that $u = x_{ij}^{2r-1}$ and $v = x_{ij}^{2r} = v_j$ for some $i, j \in \{1, \dots, n\}$ such that $v_iv_j \in E(G)$. \ReviewRevise{Since a token is placed on $u = x_{ij}^{2r-1}$, an internal vertex of the path $P_{ij}$, this token must have entered $P_{ij}$ from either $v_i$ or $v_j$. Let $w \in \{v_i, v_j\}$ denote the most recent vertex from which the token entered $P_{ij}$ before moving to $u = x_{ij}^{2r-1}$. We append the move $w \reconf[\sfTS]{G} v_j$ to $\calS$.}
				
				\item {\bf $u \in V(G^\prime) - V(G)$ and $v \in V(G^\prime) - V(G)$.} Do nothing.
			\end{enumerate}
			
			To see that $\calS$ is indeed a $\sfTS$-sequence in $G$, it suffices to show that if $C$ is the minimum vertex cover obtained right before the move $v_i \reconf[\sfTS]{G} v_j$ in $G$ then $C^\prime = C - v_i + v_j$ is also a minimum vertex cover of $G$.
			If {\bf Case 1} happens, this is trivial.
			Thus, it remains to consider the case when {\bf Case 3} happens.
			In this case, suppose to the contrary that $C^\prime$ is not a vertex cover of $G$.
			It follows that there exists $k \in \{1, \dots, n\}$ such that $v_iv_k \in E(G)$, $v_k \neq v_j$, and $v_k \notin C$.
			Intuitively, the edge $v_iv_k \in E(G)$ is not covered by any vertex in $C^\prime$.
			On the other hand, let $D$ be the D$r$DS of $G^\prime$ obtained right before the move $x_{ij}^{2r-1} \reconf[\sfTS]{G^\prime} x_{ij}^{2r} = v_j$.
			Since $D^\prime = D - x_{ij}^{2r-1} + x_{ij}^{2r}$ is also a D$r$DS of $G^\prime$, there must be some vertex in $D^\prime$ that $r$-dominates $x_{ik}^r$, which implies $V(P_{ik}) \cap D^\prime \neq \emptyset$.
			However, from the construction of $\calS$, it follows that $v_k \in C$, which is a contradiction.
			\ReviewRevise{To see that $v_k \in C$, observe that the vertex in $D^\prime$ that $r$-dominates $x_{ik}^r$ must be an internal vertex of $P_{ik}$. As a result, the token $t$ on it must have entered $P_{ik}$ from either $v_i$ or $v_k$. Furthermore, it is impossible that $t$ entered from $v_i$, since the token on $x_{ij}^{2r-1}$ is moved to $x_{ij}^{2r} = v_j$ and forms the corresponding move $v_i \reconf[\sfTS]{G} v_j$ in $\calS$. This means that the token on $x_{ij}^{2r-1}$ must have entered $P_{ij}$ from $v_i$. Note that in order to validly move the token on $v_i$ toward $x_{ij}^{2r-1}$, every vertex $x_{\ell k}^r$ where $v_\ell \in N_G(v_i) \setminus \{v_j\}$ must be $r$-dominated in $G'$. Thus, $t$ must have entered from $v_k$ (at that time, to maintain the $r$-dominating property, a token remains on $v_i$ and has not yet moved to $x_{ij}^{2r-1}$), which implies that $v_k \in C$ by our construction.}
			Thus, $C^\prime$ is a vertex cover of $G$.
			Since $|C^\prime| = |C|$, it is also minimum.
		\end{itemize}
	\end{proof}
	
	\begin{claim}\label{clm:TJ-planar}
		Under $\sfTJ$, $(G, C_s, C_t)$  is a yes-instance if and only if $(G^\prime, D_s, D_t)$ is a yes-instance.
	\end{claim}
	\begin{proof}
		\begin{itemize}
			\item[($\Rightarrow$)] Let $\calS$ be a $\sfTJ$-sequence in $G$ between $C_s$ and $C_t$.
			Since any minimum vertex cover of $G$ is also a minimum D$r$DS of $G^\prime$, the sequence $\calS^\prime$ obtained by replacing each move $u \reconf[\sfTJ]{G} v$ in $\calS$ by $u \reconf[\sfTJ]{G^\prime} v$ is also a $\sfTS$-sequence in $G^\prime$ between $D_s = C_s$ and $D_t = C_t$.
			
			\item[($\Leftarrow$)] Let $\calS^\prime$ be a $\sfTJ$-sequence in $G^\prime$ between $D_s$ and $D_t$.
			We construct a sequence of token-jumps $\calS$ in $G$ between $C_s = D_s$ and $C_t = D_t$ as follows.
			Initially, $\calS = \emptyset$.
			For each move $u \reconf[\sfTJ]{G^\prime} v$ in $\calS^\prime$, we consider the following cases.
			\begin{enumerate}[{\bf {Case} 1:}]
				\item {\bf $u \in V(G)$ and $v \in V(G)$.} It must happen that $u = v_i$ and $v = v_j$ for some $i, j \in \{1, \dots, n\}$. We append $v_i \reconf[\sfTJ]{G} v_j$ to $\calS$.
				
				\item {\bf $u \in V(G)$ and $v \in V(G^\prime) - V(G)$.} Do nothing.
				
				\item {\bf $u \in V(G^\prime) - V(G)$ and $v \in V(G)$.} From the construction of $G^\prime$, for each pair $i, j \in \{1, \dots, n\}$ such that $v_iv_j \in E(G)$, the vertex $x_{ij}^r$ must be $r$-dominated by at least one vertex of $\mathcal{C}_{ij}$.
				Additionally, note that any token-set resulting from a move in $\calS^\prime$ must be a minimum D$r$DS of $G^\prime$.
				Thus, we must have $u = x_{ij}^p$ and $v = x_{ij}^{2r} = v_j$ for some $i, j \in \{1, \dots, n\}$ such that $v_iv_j \in E(G)$ and $1 \leq p \leq 2r-1$.	
				Now, we append $v_i \reconf[\sfTJ]{G} v_j$ to $\calS$.
				
				\item {\bf $u \in V(G^\prime) - V(G)$ and $v \in V(G^\prime) - V(G)$.} Do nothing.
			\end{enumerate}
			
			To see that $\calS$ is indeed a $\sfTJ$-sequence in $G$, it suffices to show that if $C$ is the minimum vertex cover obtained right before the move $v_i \reconf[\sfTJ]{G} v_j$ in $G$ then $C^\prime = C - v_i + v_j$ is also a minimum vertex cover of $G$.
			If {\bf Case 1} happens, this is trivial.
			Thus, it remains to consider the case {\bf Case 3} happens.
			In this case, suppose to the contrary that $C^\prime$ is not a vertex cover of $G$.
			It follows that there exists $k \in \{1, \dots, n\}$ such that $v_iv_k \in E(G)$, $v_k \neq v_j$, and $v_k \notin C$.
			Intuitively, the edge $v_iv_k \in E(G)$ is not covered by any vertex in $C^\prime$.
			On the other hand, let $D$ be the D$r$DS of $G^\prime$ obtained right before the move $x_{ij}^{p} \reconf[\sfTJ]{G^\prime} x_{ij}^{2r} = v_j$, for $1 \leq p \leq 2r-1$.
			Since $D^\prime = D - x_{ij}^{p} + x_{ij}^{2r}$ is also a D$r$DS of $G^\prime$, there must be some vertex in $D^\prime$ that $r$-dominates $x_{ik}^r$, which implies $V(P_{ik}) \cap D^\prime \neq \emptyset$.
			However, from the construction of $\calS$, it follows that $v_k \in C$, which is a contradiction.
			Thus, $C^\prime$ is a vertex cover of $G$.
			Since $|C^\prime| = |C|$, it is also minimum.
		\end{itemize}
	\end{proof}
	Our proof is complete.
\end{proof}

In the following theorem, we show that we can further improve the above result by reducing from the NCL problem~\cite{HearnD05,Zanden15}.
\begin{theorem}\label{thm:planar-maxdeg3}
	\textsc{D$r$DSR} under $\sfR \in \{\sfTS, \sfTJ\}$ on planar graphs of maximum degree three and bounded bandwidth is $\ttPSPACE$-complete for any $r \geq 1$.
\end{theorem}
\begin{proof}
	We reduce from the NCL problem, which is known to be $\ttPSPACE$-complete on planar graphs of bounded bandwidth~\cite{HearnD05,Zanden15}.
	We first consider $\sfTS$ and later explain how the same reduction holds under $\sfTJ$.
	The \textsc{And} and \textsc{Or} gadgets that simulate the behaviors of the NCL \textsc{And} and \textsc{Or} vertices, respectively, are constructed as in \figurename~\ref{fig:NCL-planar}. 
	\begin{figure}[ht]
		\centering
		\includegraphics[width=0.7\textwidth]{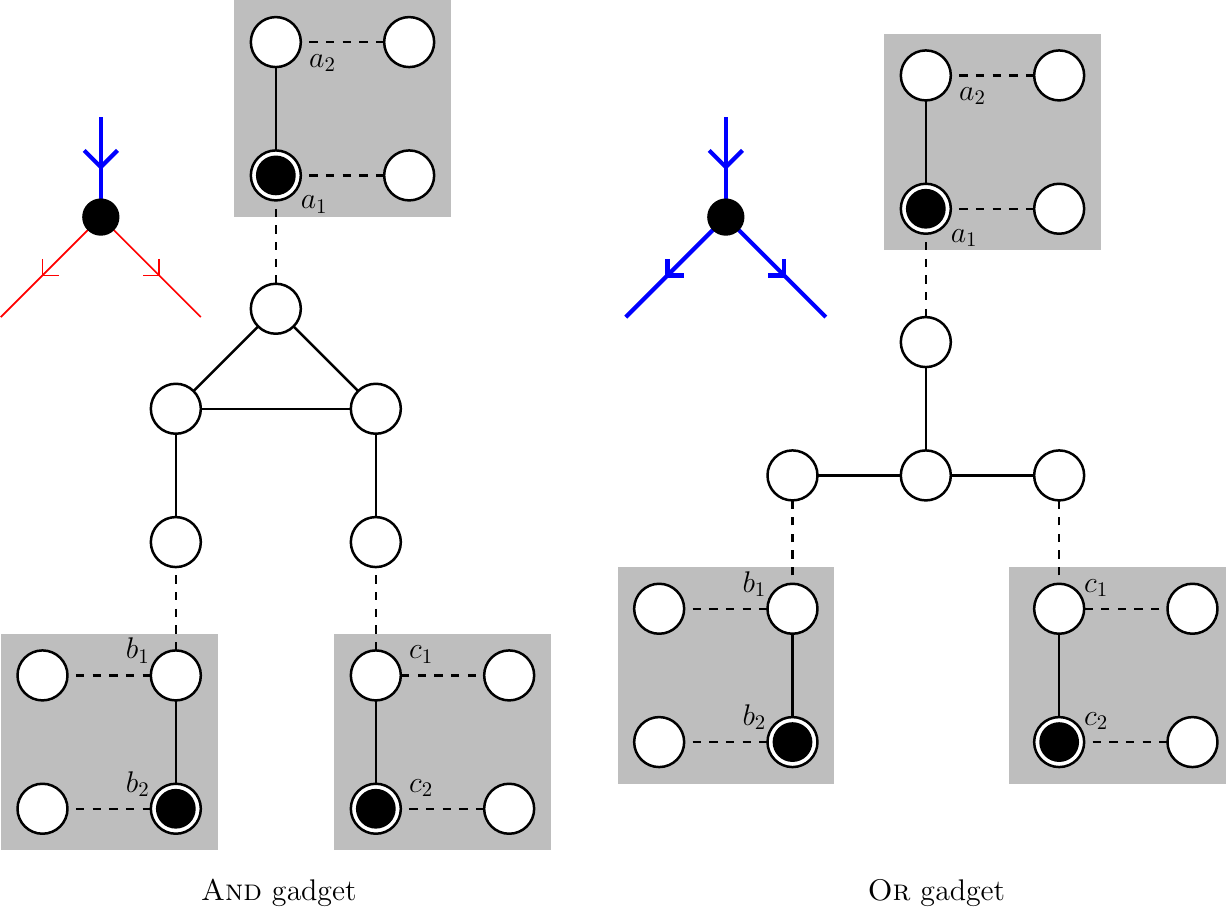}
		\caption{Our gadgets for \textsc{D$r$DSR}. Each dashed edge represents a path of length $r-1$.}
		\label{fig:NCL-planar}
	\end{figure}
	\begin{itemize}
		\item A gadget has exactly one \textit{main component} and three disjoint \textit{link components}. 
		\item The \textit{main component} of the \textsc{And} gadget contains a bull graph (i.e., a triangle with two disjoint pendant edges). 
		There is exactly one path of length $r-1$ attached to each vertex of degree at most two of the graph. 
		The \textit{main component} of the \textsc{Or} gadget contains a claw graph (i.e., a $K_{1,3}$). 
		There is exactly one path of length $r-1$ attached to each leaf of the claw graph.
		(See \figurename~\ref{fig:NCL-planar}.)
		\item Each \textit{link component} in both gadgets is formed by attaching paths of length $r-1$ to the endpoints of each edge: $\{(a_1a_2), (b_1b_2), (c_1c_2)\}$. 
		More precisely, there is exactly one path of length $r-1$ starting from each of the vertices $a_1$, $a_2$, $b_1$, $b_2$, $c_1$, and $c_2$.
		In \figurename~\ref{fig:NCL-planar}, each link component is shown inside a gray box.
  		\item In each gadget, its main component and three link components are joined together respectively at the vertices $a_1$, $b_1$, and $c_1$, as in \figurename~\ref{fig:NCL-planar}.
	\end{itemize}
  
		We want these link components to simulate the directions of NCL edges as follows: when a token is placed on either $a_1$, $b_1$, or $c_1$, the corresponding NCL edge is directed inward, and when a token is placed on either $a_2$, $b_2$, or $c_2$, the corresponding NCL edge is directed outward.
		Naturally, \textit{sliding a token}, say from $a_1$ to $a_2$, corresponds to \textit{reversing the direction of the corresponding NCL edge}, say from inward to outward, and vice versa.
		\figurename~\ref{fig:NCL-planar} illustrates this setting:
		\begin{itemize}
			\item The current token-set of the \textsc{And} gadget corresponds to a \textsc{And} vertex where the weight-$2$ (blue) edge is directed inward and the two weight-$1$ (red) edges are directed outward.
			\item The current token-set of the \textsc{Or} gadget corresponds to a \textsc{Or} vertex where one weight-$2$ (blue, top) edge is directed inward and the other two weight-$2$ (blue, bottom) edges are directed outward.
		\end{itemize}

	Given an NCL graph, which is planar and bounded bandwidth, and one of its configurations, we construct a corresponding graph and token-set by joining together \textsc{And} and \textsc{Or} gadgets at their shared link components and placing the tokens in link components appropriately.
	\figurename~\ref{fig:NCL-TS-graph} describes the constructed graph and token-set corresponds to the NCL graph and its configuration given in \figurename~\ref{fig:NCL-def}(a).
	One can verify that the constructed token-set indeed forms a D$r$DS of the constructed graph. (It suffices to verify for each gadget.)
	For convenience, we call each constructed token-set a \textit{standard token placement} and each token in a link component a \textit{link token}.

	\begin{figure}[ht]
		\centering
		\includegraphics[width=0.7\textwidth]{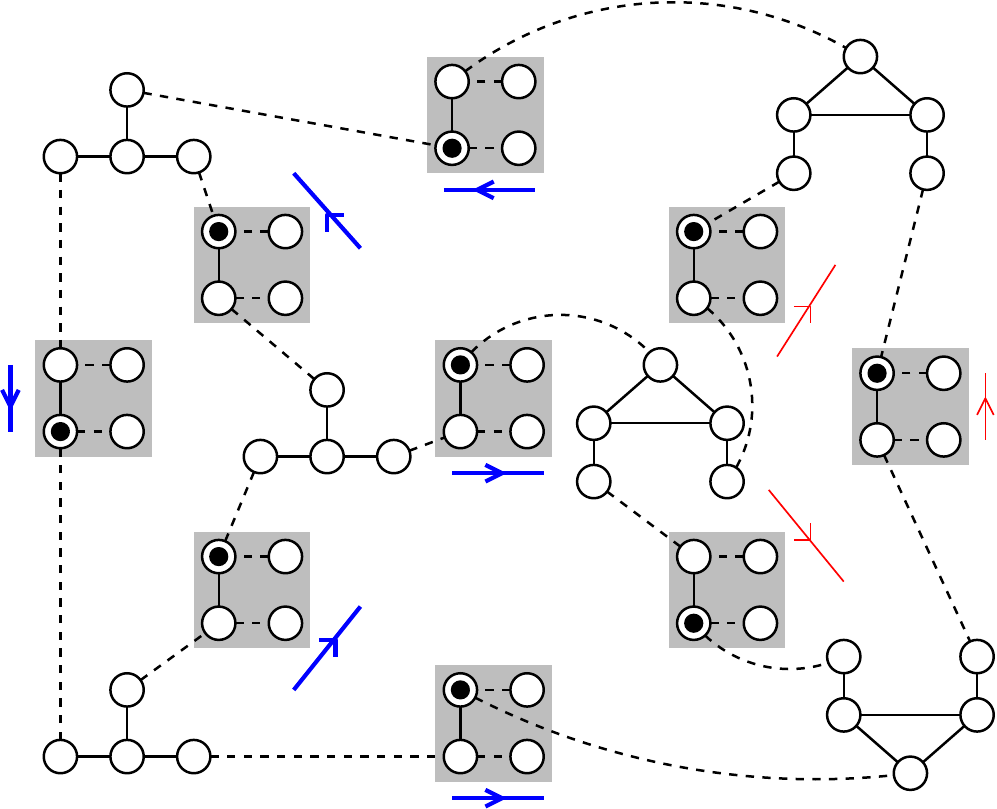}
		\caption{The graph and token-set corresponding to the NCL \textsc{And/or} constraint graph and its configuration in \figurename~\ref{fig:NCL-def}(a). Each dashed edge represents a path of length $r-1$.}
		\label{fig:NCL-TS-graph}
	\end{figure}

	We now show the correctness of our construction. 
	Observe that, starting from a standard token placement, no link token ever leaves its corresponding edge ($(a_1a_2)$, $(b_1b_2)$, or $(c_1c_2)$). 
	Otherwise, one of the two degree-$1$ vertices in the link component would not be $r$-dominated by the resulting token-set.
	Hence, a one-to-one mapping exists between the direction of an NCL edge and the placement of a link token: a link token is placed on either $a_1$, $b_1$, or $c_1$ if and only if the corresponding NCL edge is directed inward, and on either $a_2$, $b_2$, or $c_2$ if and only if the corresponding NCL edge is directed outward.

	The \textsc{And} gadget satisfies the same constraint as an NCL \textsc{And} vertex. 
	The top token can slide from $a_1$ to $a_2$ (which corresponds to the reversal of the top weight-$2$ (blue) edge from inward to outward) only when both bottom tokens slide in respectively from $b_2$ to $b_1$ and $c_2$ to $c_1$ (corresponding to the reversal of the two bottom weight-$1$ (red) edges from outward to inward).

	Similarly, the \textsc{Or} gadget satisfies the same constraint as an NCL \textsc{Or} vertex. The top token can slide from $a_1$ to $a_2$ (corresponding to the reversal of the top weight-$2$ (blue) edge from inward to outward) only when at least one of the two bottom tokens slides either from $b_2$ to $b_1$ or from $c_2$ to $c_1$ (corresponding to the reversal of at least one of the two bottom weight-$2$ (blue) edges from outward to inward).

	Thus, given a sequence $\langle C_s = C_1, C_2, \dots, C_p = C_t \rangle$ of NCL configurations where $C_i$ and $C_{i+1}$ are adjacent, for $1 \leq i \leq p-1$, one can construct (using the above one-to-one mapping) a $\sfTS$-sequence between the corresponding standard token placements, and vice versa.
	Therefore, our reduction is correct.

	Recall that the original NCL graph is planar and bounded bandwidth.
	Since both our gadgets are planar, having only a constant number of edges (as $r$ is a fixed constant), and of maximum degree three, 
	the constructed graph is also planar, bounded bandwidth, and of maximum degree three.
	(For example, see \figurename~\ref{fig:NCL-TS-graph}).
	Finally, as link tokens never leave their corresponding edges, any $\sfTJ$-move is also a $\sfTS$-move, and therefore our proof holds under $\sfTJ$.
\end{proof}

\subsection{Chordal Graphs}
\label{sec:chordal}

\begin{theorem}\label{thm:chordal}
	\textsc{D$r$DSR} under $\sfR \in \{\sfTS, \sfTJ\}$ on chordal graphs is $\ttPSPACE$-complete for any $r \geq 2$.
\end{theorem}
\begin{proof}
	We give a polynomial-time reduction from \textsc{Min-VCR} on general graphs.
	We remark that a similar idea has been used in~\cite{Hoang23}.
	Let $(G, C_s, C_t)$ be an instance of \textsc{Min-VCR} under $\sfR$ where $C_s, C_t$ are two minimum vertex covers of a graph $G$.
	We will construct an instance $(G^\prime, D_s, D_t)$ of \textsc{D$r$DSR} under $\sfR$ where $D_s$ and $D_t$ are two D$r$DSs of a chordal graph $G^\prime$.
	
	Suppose that $V(G) = \{v_1, \dots, v_n\}$.
	We construct $G^\prime$ from $G$ as follows.
	\begin{itemize}
		\item Form a clique in $G^\prime$ having all vertices $v_1, \dots, v_n$ of $G$
		\item For each edge $v_iv_j \in E(G)$ ($i, j \in \{1, \dots, n\}$) and for $p \in \{1, 2\}$, add a corresponding new vertex $x^p_{ij}$ and create an edge from it to both $v_i$ and $v_j$.
		Observe that $x^p_{ij} = x^p_{ji}$.
		Furthermore, create a new path $P^p_{ij}$ of length exactly $r-1$ \ReviewRevise{having $x^p_{ij}$ as an endpoint}.  
	\end{itemize}
	We define $D_s = C_s$ and $D_t = C_t$. 
	Clearly, this construction can be done in polynomial time.
	(See \figurename~\ref{fig:chordal}.)
	
	\begin{figure}[ht]
		\centering
		\includegraphics[width=0.7\textwidth]{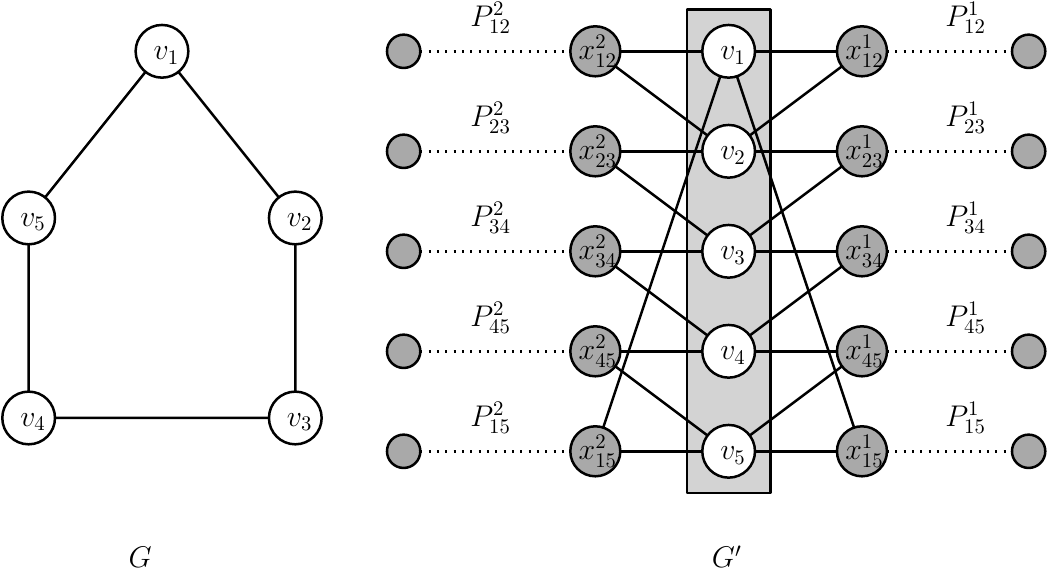}
		\caption{An example of constructing $G^\prime$ from a graph $G$ in the proof of Theorem~\ref{thm:chordal}. Vertices in $V(G^\prime) - V(G)$ are marked with the gray color. Vertices in the light gray box form a clique. Each dotted path is of length exactly $r-1$.}
		\label{fig:chordal}
	\end{figure}
	
	Now, we show that $(G^\prime, D_s, D_t)$ is indeed our desired \textsc{D$r$DSR}'s instance.
	It follows from the construction that $G^\prime$ is indeed a chordal graph.
	More precisely, if we define $H = (K \uplus S, F)$ to be the subgraph of $G^\prime$ induced by $K \cup S$ with $K = \{v_1, \dots, v_n\}$ forming a clique and $S = \bigcup_{p = 1}^2\bigcup_{\{i, j \mid v_iv_j \in E(G)\}}\{x_{ij}^p\}$ forming an independent set, then $H$ is indeed a split graph.
	Moreover, $G^\prime$ is obtained from $H$ by attaching paths to each member of $S$, which clearly results a chordal graph.
	Additionally, one can verify that any minimum vertex cover of $G$ is also a minimum dominating set of $H$ and therefore a minimum D$r$DS of $G^\prime$.
	
	It remains to show that our reduction is correct, i.e., under $\sfR \in \{\sfTS, \sfTJ\}$, $(G, C_s, C_t)$  is a yes-instance if and only if $(G^\prime, D_s, D_t)$ is a yes-instance.
	\begin{itemize}
		\item[($\Rightarrow$)] Let $\calS$ be a $\sfR$-sequence in $G$ between $C_s$ and $C_t$.
		Since any minimum vertex cover of $G$ is also a minimum D$r$DS of $G^\prime$, the sequence $\calS^\prime$ obtained by replacing each move $u \reconf[\sfR]{G} v$ in $\calS$ by $u \reconf[\sfR]{G^\prime} v$ is also a $\sfR$-sequence in $G^\prime$ between $D_s = C_s$ and $D_t = C_t$.
		
		\item[($\Leftarrow$)] Let $\calS^\prime$ be a $\sfR$-sequence in $G^\prime$ between $D_s$ and $D_t$.
		From the construction of $G^\prime$, observe that no token can be moved to a vertex in $V(G^\prime) - V(G)$; otherwise some degree-$1$ endpoint of a $P^p_{ij}$ ($p \in \{1, 2\}$, $i, j \in \{1, \dots, n\}$) would not be $r$-dominated by the resulting token-set.
		Therefore, any move $u \reconf[\sfR]{G^\prime} v$ in $\calS^\prime$ satisfies that both $u$ and $v$ are in $V(G)$, and thus can be replaced by the move $u \reconf[\sfR]{G} v$ to construct $\calS$---our desired $\sfR$-sequence between $C_s = D_s$ and $C_t = D_t$ in $G$.
	\end{itemize}
	Our proof is complete.
\end{proof}

\subsection{Bipartite Graphs}
\label{sec:bipartite}

\begin{theorem}\label{thm:bipartite}
	\textsc{D$r$DSR} under $\sfR \in \{\sfTS, \sfTJ\}$ on bipartite graphs is $\ttPSPACE$-complete for any $r \geq 2$.
\end{theorem}
\begin{proof}
	We give a polynomial-time reduction from \textsc{Min-VCR} on general graphs.
	Our reduction extends the one given by Bonamy et al. \cite{BonamyDO21} for the case $r = 1$.
	Let $(G, C_s, C_t)$ be an instance of \textsc{Min-VCR} under $\sfR$ where $C_s, C_t$ are two minimum vertex covers of a graph $G$.
	We will construct an instance $(G^\prime, D_s, D_t)$ of \textsc{D$r$DSR} under $\sfR$ where $D_s$ and $D_t$ are two D$r$DSs of a bipartite graph $G^\prime$.
	
	Suppose that $V(G) = \{v_1, \dots, v_n\}$.
	We construct $G^\prime$ from $G$ as follows.
	\begin{enumerate}[(a)]
		\item Replace each edge $v_iv_j$ by a path $P_{ij} = x_{ij}^0x_{ij}^1\dots x_{ij}^{2r}$ of length $2r$ ($1 \leq i, j \leq n$) with $x_{ij}^0 = v_i$ and $x_{ij}^{2r} = v_j$. 
		Observe that $x_{ij}^p = x_{ji}^{2r - p}$ for $0 \leq p \leq 2r$.
		\item Add a new vertex $x$ and join it to every vertex in $V(G)$.
		\item Attach a new path $P_x$ of length $r$ to $x$.
	\end{enumerate}
	We define $D_s = C_s + x$ and $D_t = C_t + x$.
	Clearly, this construction can be done in polynomial time.
	(See \figurename~\ref{fig:bipartite}.)
	
	\begin{figure}[ht]
		\centering
		\includegraphics[width=0.8\textwidth]{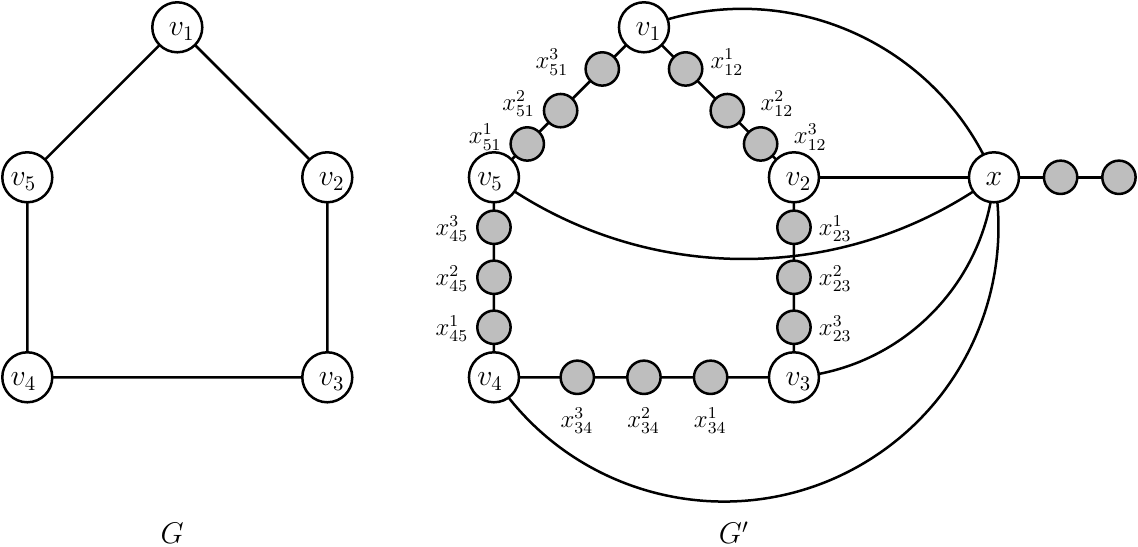}
		\caption{An example of constructing a bipartite graph $G^\prime$ from a graph $G$ when $r = 2$ in the proof of Theorem~\ref{thm:bipartite}.}
		\label{fig:bipartite}
	\end{figure}
	
	In the next two claims, we prove that our construction results in an instance of \textsc{D$r$DSR} on bipartite graphs: Claim~\ref{clm:Gp-bipartite} shows that $G^\prime$ is bipartite and Claim~\ref{clm:Cx-DrDS} implies that both $D_s$ and $D_t$ are minimum D$r$DSs of $G^\prime$.
	\begin{claim}\label{clm:Gp-bipartite}
		$G^\prime$ is a bipartite graph.
	\end{claim}
	\begin{proof}
		We show that any cycle $\mathcal{C}^\prime$ in $G^\prime$ has even length and therefore $G^\prime$ is bipartite.
		Observe that no cycle of $G^\prime$ contains a vertex from $V(P_x) - x$.
		From the construction, if $\mathcal{C}^\prime$ does not contain $x$, it follows that $V(\mathcal{C}^\prime) \cap V(G)$ must form a cycle $\mathcal{C}$ of $G$.
		Therefore, $\mathcal{C}^\prime$ is of length $2r|E(\mathcal{C})|$, which is even.
		On the other hand, if $\mathcal{C^\prime}$ contains $x$, it follows that $V(\mathcal{C}^\prime) \cap V(G)$ must form a path $\mathcal{P}$ of $G$.
		Therefore, $\mathcal{C}^\prime$ is of length $2 + 2r|E(\mathcal{P})|$, which again is even.
	\end{proof}
	
	\begin{claim}\label{clm:Cx-DrDS}
		Any set of the form $C + x$, where $C$ is a minimum vertex cover of $G$, is a minimum D$r$DS of $G^\prime$.
	\end{claim}
	\begin{proof}
		To see this, note that, by construction, in $G^\prime$, $x$ $r$-dominates every vertex in $V(G^\prime) - \bigcup_{v_iv_j \in E(G)}\{x_{ij}^r\}$.
		Additionally, since each $x_{ij}^r$ belongs to exactly one path $P_{ij}$ and $C$ is a minimum vertex cover of $G$, it follows that in $G^\prime$, $C$ $r$-dominates $\bigcup_{v_iv_j \in E(G)}\{x_{ij}^r\}$.
		Thus, $C + x$ $r$-dominates $V(G^\prime)$, i.e., $C + x$ is a D$r$DS of $G^\prime$.
		
		It remains to show that $C + x$ is minimum.
		Indeed, it is sufficient to show that $\tau(G) + 1 = \gamma_r(G^\prime)$ where $\tau(G)$ and $\gamma_r(G^\prime)$ are respectively the size of a minimum vertex cover of $G$ and a minimum D$r$DS of $G^\prime$.
		Since $C + x$ is a D$r$DS of $G^\prime$, we have $\tau(G) + 1 = |C| + x \geq  \gamma_r(G^\prime)$.
		On the other hand, note that any minimum D$r$DS $D^\prime$ of $G^\prime$ must $r$-dominate $V(P_x)$ and therefore contains a vertex of $P_x$ (which, by the construction, does not belong to $G$).
		Moreover, from the construction of $G^\prime$, each path $P_{ij}$ ($1 \leq i, j \leq n$ and $v_iv_j \in E(G)$) is of length $2r$.
		Thus, in order to $r$-dominate all $V(P_{ij})$, $D^\prime$ needs to contain at least one vertex from each path $P_{ij}$.
		Therefore, $\gamma_r(G^\prime) = |D^\prime| \geq 1 + \tau(G)$.
		Our proof is complete.
	\end{proof}
	
	Before proving the correctness of our reduction, we prove some useful observations.
	\begin{claim}\label{clm:vertex-x}
		Let $D = C + x$ be a D$r$DS of $G^\prime$ where $C$ is a minimum vertex cover of $G$.
		Then,
		\begin{enumerate}[(a)]
			\item $D - u + y$ is not a D$r$DS of $G^\prime$, for any $u \in C$ and $y \in V(P_x) - x$.
			\item $D - x + v$ is not a D$r$DS of $G^\prime$, for any $v \in V(G^\prime) - x$.
			\item $D - u + z$ is not a D$r$DS of $G^\prime$, for $u = v_i \in C$ and $z \notin \bigcup_{\{j \mid v_j \in N_G(v_i) - C\}}V(P_{ij})$.
		\end{enumerate}
	\end{claim}
	\begin{proof}
		\begin{enumerate}[(a)]
			\item Suppose to the contrary there exists $u \in C$ and $y \in V(P_x) - x$ such that $D^\prime = D - u + y$ is a D$r$DS of $G^\prime$.
			Since $|D^\prime| = |D| = \tau(G) + 1$, Claim~\ref{clm:Cx-DrDS} implies that $D^\prime$ is a minimum D$r$DS of $G^\prime$.
			On the other hand, from the construction of $G^\prime$, any vertex $r$-dominated by $y$ must also be $r$-dominated by $x$. 
			Thus, $D^\prime - y$ is also a D$r$DS of $G^\prime$, which is a contradiction.
			
			\item Suppose to the contrary there exists $v \in V(G^\prime) - x$ such that $D^\prime = D - x + v$ is a D$r$DS of $G^\prime$.
			From (a), it follows that $v \in V(P_x) - x$; otherwise some vertex of $P_x$ would not be $r$-dominated by any member of $D^\prime$.
			Since $C$ is a minimum vertex cover of $G$, it follows that there exists a pair $i, j \in \{1, \dots, n\}$ such that $v_iv_j \in E(G)$ and $C \cap \{v_i, v_j\} = \{v_i\}$; otherwise every vertex of $G$ would contain a token in $C$ and therefore $C$ would not be minimum---a contradiction.
			From the construction of $G^\prime$, $x$ is the unique vertex in $V(P_x)$ that $r$-dominates $x_{ij}^{r+1}$.
			Thus, $x_{ij}^{r+1}$ is not $r$-dominated by any vertex in $D^\prime = D - x + v$ for $v \in V(P_x) - x$, which is a contradiction.
			
			\item Let $j \in \{1, \dots, n\}$ be such that $v_j \in N_G(u) - C = N_G(v_i) - C$.
			Since $C$ is a minimum vertex cover of $G$, such a vertex $v_j$ must exist.
			From the construction of $G^\prime$, the vertex $u = v_i$ is the unique vertex in $D$ that $r$-dominates $x_{ij}^r$.
			Thus, in order to keep $x_{ij}^r$ being $r$-dominated in $G^\prime$, a token on $u$ can only be moved to some vertex in $\bigcup_{\{j \mid v_j \in N_G(v_i) - C\}}V(P_{ij})$.
		\end{enumerate}
	\end{proof}
	Intuitively, starting from a token-set of the form $C + x$ for some minimum vertex cover $C$ of $G$, Claim~\ref{clm:vertex-x}(a) means that as long as $x$ has a token, no other token can be moved in $G^\prime$ to a vertex in $P_x - x$, Claim~\ref{clm:vertex-x}(b) implies that if $x$ has a token, it will never be moved in $G^\prime$, and Claim~\ref{clm:vertex-x}(c) says that a token cannot be moved in $G^\prime$ ``too far'' from its original position.
	
	We are now ready to show the correctness of our reduction.
	(Claims~\ref{clm:TS-bipartite} and~\ref{clm:TJ-bipartite}.)
	\begin{claim}\label{clm:TS-bipartite}
		Under $\sfTS$, $(G, C_s, C_t)$  is a yes-instance if and only if $(G^\prime, D_s, D_t)$ is a yes-instance.
	\end{claim}
	\begin{proof}
		\begin{itemize}
			\item[($\Rightarrow$)] Suppose that $\calS$ is a $\sfTS$-sequence in $G$ between $C_s$ and $C_t$.
			We construct a sequence $\calS^\prime$ of token-slides in $G^\prime$ between $D_s$ and $D_t$ by replacing each $v_i \reconf[\sfTS]{G} v_j$ in $\calS$ with the sequence $\calS_{ij} = \langle v_i \reconf[\sfTS]{G^\prime} x_{ij}^1, x_{ij}^1 \reconf[\sfTS]{G^\prime} x_{ij}^2, \dots, x_{ij}^{2r-1} \reconf[\sfTS]{G^\prime} x_{ij}^{2r} \rangle$, where $i, j \in \{1, \dots, n\}$ and $v_iv_j \in E(G)$.
			Intuitively, sliding a token $t_{ij}$ from $v_i$ to its neighbor $v_j$ in $G$ corresponds to sliding $t_{ij}$ from $v_i$ to $v_j$ in $G^\prime$ along the path $P_{ij}$.
			Since $x$ always $r$-dominates $V(G^\prime) - \bigcup_{v_iv_j \in E(G)}\{x_{ij}^r\}$ and after each move in $\calS_{ij}$ the token $t_{ij}$ always $r$-dominates $x_{ij}^r$, it follows that $\calS_{ij}$ is indeed a $\sfTS$-sequence in $G^\prime$.
			Thus, $\calS^\prime$ is our desired $\sfTS$-sequence.
			
			\item[($\Leftarrow$)] Let $\calS^\prime$ be a $\sfTS$-sequence in $G^\prime$ between $D_s$ and $D_t$.
			We describe how to construct a $\sfTS$-sequence $\calS$ in $G$ between $C_s$ and $C_t$.
			Initially, $\calS = \emptyset$.
			For each move $u \reconf[\sfTS]{G^\prime} v$ ($u \neq v$) in $\calS^\prime$, we consider the following cases:
			\begin{enumerate}[{\bf {Case} 1:}]
				\item {\bf Either $u \in V(P_x)$ or $v \in V(P_x)$.} 
				Since $\calS^\prime$ starts with $D_s = C_s + x$, Claim~\ref{clm:vertex-x} ensures that no token can be placed on some vertex in $P_x - x$ and $x$ always contains a token.
				As a result, this case does not happen.
				
				\item {\bf $u = x_{ij}^p$ and $v = x_{ij}^{p+1}$ for $0 \leq p \leq 2r$.} 
				Recall that in $G^\prime$, we have $x_{ij}^0 = v_i$ and $x_{ij}^{2r} = v_j$. 
				Additionally, from our construction, since $x_{ij}^p \in V(G^\prime)$ has a token, so does $v_i \in V(G)$.
				\ReviewRevise{If $p = 2r-1$, we append the move $w \reconf[\sfTS]{G} v_j$ to $\calS$, where $w \in \{v_i, v_j\}$ is the most recent vertex from which the token on $x^{2r - 1}_{ij}$ entered $P_{ij}$.}
				Otherwise, we do nothing.
			\end{enumerate}
			To see that $\calS$ is indeed a $\sfTS$-sequence in $G$, it suffices to show that if $C$ is the minimum vertex cover obtained right before the move $v_i \reconf[\sfTS]{G} v_j$ then $C^\prime = C - v_i + v_j$ is also a minimum vertex cover of $G$.
			Suppose to the contrary that $C^\prime$ is not a vertex cover of $G$.
			It follows that there exists $k \in \{1, \dots, n\}$ such that $v_iv_k \in E(G)$, $v_k \neq v_j$, and $v_k \notin C$.
			Intuitively, the edge $v_iv_k$ is not covered by any vertex in $C^\prime$.
			On the other hand, let $D$ be the D$r$DS of $G^\prime$ obtained right before the move $x_{ij}^{2r-1} \reconf[\sfTS]{G^\prime} x_{ij}^{2r} = v_j$.
			Since $D^\prime = D - x_{ij}^{2r-1} + x_{ij}^{2r}$ is also a D$r$DS of $G^\prime$, there must be some vertex in $D^\prime$ that $r$-dominates $x_{ik}^r$, which implies $V(P_{ik}) \cap D^\prime \neq \emptyset$.
			However, from the construction of $\calS$, it follows that $v_k \in C$, which is a contradiction.
			\ReviewRevise{To see that $v_k \in C$, observe that the vertex in $D^\prime$ that $r$-dominates $x_{ik}^r$ must be an internal vertex of $P_{ik}$. As a result, the token $t$ on it must have entered $P_{ik}$ from either $v_i$ or $v_k$. Furthermore, it is impossible that $t$ entered from $v_i$, since the token on $x_{ij}^{2r-1}$ is moved to $x_{ij}^{2r} = v_j$ and forms the corresponding move $v_i \reconf[\sfTS]{G} v_j$ in $\calS$. This means that the token on $x_{ij}^{2r-1}$ must have entered $P_{ij}$ from $v_i$. Note that in order to validly move the token on $v_i$ toward $x_{ij}^{2r-1}$, every vertex $x_{\ell k}^r$ where $v_\ell \in N_G(v_i) \setminus \{v_j\}$ must be $r$-dominated in $G'$. Thus, $t$ must have entered from $v_k$ (at that time, to maintain the $r$-dominating property, a token remains on $v_i$ and has not yet moved to $x_{ij}^{2r-1}$), which implies that $v_k \in C$ by our construction.}
			Thus, $C^\prime$ is a vertex cover of $G$.
			Since $|C^\prime| = |C|$, it is also minimum.
		\end{itemize}
	\end{proof}
	
	\begin{claim}\label{clm:TJ-bipartite} %
		Under $\sfTJ$, $(G, C_s, C_t)$  is a yes-instance if and only if $(G^\prime, D_s, D_t)$ is a yes-instance.
	\end{claim}
	\begin{proof}
		\begin{itemize}
			\item[($\Rightarrow$)] Suppose that $\calS$ is a $\sfTJ$-sequence in $G$ between $C_s$ and $C_t$.
			It follows from Claim~\ref{clm:Cx-DrDS} that the sequence $\calS^\prime$ of token-jumps obtained from $\calS$ by replacing each move $u \reconf[\sfTJ]{G} v$ in $\calS$ by $u \reconf[\sfTJ]{G^\prime} v$ is a $\sfTJ$-sequence between $D_s = C_s + x$ and $D_t = C_t + x$.
			
			\item[($\Leftarrow$)] On the other hand, let $\calS^\prime$ be a $\sfTJ$-sequence in $G^\prime$ between $D_s$ and $D_t$.
			We describe how to construct a $\sfTJ$-sequence $\calS$ in $G$ between $C_s$ and $C_t$.
			Initially, $\calS = \emptyset$.
			For each move $u \reconf[\sfTJ]{G^\prime} v$ ($u \neq v$) in $\calS^\prime$, we consider the following cases:
			\begin{enumerate}[{\bf {Case} 1:}]
				\item {\bf Either $u \in V(P_x)$ or $v \in V(P_x)$.} As before, it follows from Claim~\ref{clm:vertex-x} that this case does not happen.
				
				\item {\bf $u = x_{ij}^p$ and $v = x_{ij}^{q}$ for $0 \leq p, q \leq 2r$.} 
				Recall that in $G^\prime$ we have $x_{ij}^0 = v_i$ and $x_{ij}^{2r} = v_j$. 
				If $q = 2r$, we append the move $v_i \reconf[\sfTJ]{G} v_j$ to $\calS$.
				Otherwise, we do nothing.
				
				\item {\bf $u = x_{ij}^p$ and $v = x_{k\ell}^{q}$ for two edges $v_iv_j$ and $v_kv_\ell$ in $G$ and $0 \leq p, q \leq 2r$.}
				Note that if ($\star$)~$v_iv_j$ and $v_kv_\ell$ are adjacent edges in $G$ and either $u$ or $v$ is their common endpoint, we are back to {\bf Case~2}.
				Thus, let's assume ($\star$) does not happen and consider the \textit{first} move of this type.
				From Claim~\ref{clm:vertex-x}, it must happen that before the move $x_{ij}^p \reconf[\sfTJ]{G^\prime} x_{k\ell}^{q}$, some move in $\calS^\prime$ places a token on $u = x_{ij}^p$ and moreover such a token must come from either $v_i$ or $v_j$.
				Additionally, again by Claim~\ref{clm:vertex-x}, if the token comes from $v_i$ then $v_j$ contains no other token, and vice versa.
				Let $D$ be the D$r$DS obtained right before the move $x_{ij}^p \reconf[\sfTJ]{G^\prime} x_{k\ell}^{q}$.
				Now, since $x_{ij}^p \reconf[\sfTJ]{G^\prime} x_{k\ell}^{q}$ is the first move of its type, it follows that the path $P_{ij}$ contains exactly one token from $D$ which is placed on $x_{ij}^p$.
				However, this means we cannot perform $x_{ij}^p \reconf[\sfTJ]{G^\prime} x_{k\ell}^{q}$; otherwise $x_{ij}^r$ would not be $r$-dominated by the resulting token-set.
				Thus, this case does not happen unless ($\star$) is satisfied.
			\end{enumerate}
			To see that $\calS$ is indeed a $\sfTJ$-sequence in $G$, it suffices to show that if $C$ is the minimum vertex cover obtained right before the move $v_i \reconf[\sfTJ]{G} v_j$ then $C^\prime = C - v_i + v_j$ is also a minimum vertex cover of $G$.
			Suppose to the contrary that $C^\prime$ is not a vertex cover of $G$.
			It follows that there exists $k \in \{1, \dots, n\}$ such that $v_iv_k \in E(G)$, $v_k \neq v_j$, and $v_k \notin C$.
			Intuitively, the edge $v_iv_k \in E(G)$ is not covered by any vertex in $C^\prime$.
			On the other hand, let $D$ be the D$r$DS of $G^\prime$ obtained right before the move $x_{ij}^p \reconf[\sfTJ]{G^\prime} x_{ij}^{2r} = v_j$.
			Since $D^\prime = D - x_{ij}^{p} + x_{ij}^{2r}$ is also a D$r$DS of $G^\prime$, there must be some vertex in $D^\prime$ that $r$-dominates $x_{ik}^r$, which implies $V(P_{ik}) \cap D^\prime \neq \emptyset$.
			However, from the construction of $\calS$, it follows that $v_k \in C$, which is a contradiction.
			Thus, $C^\prime$ is a vertex cover of $G$.
			Since $|C^\prime| = |C|$, it is also minimum.
		\end{itemize}
	\end{proof}
	Our proof is complete.
\end{proof}

\section{Concluding Remarks}
\label{sec:conclude}

In this paper, we provide an initial picture of the computational complexity of \textsc{D$r$DSR} ($r \geq 2$) under $\sfTS$ and $\sfTJ$ on different graph classes.
We extended several known results for $r = 1$ and provided a complexity dichotomy of \textsc{D$r$DSR} on split graphs: the problem is $\ttPSPACE$-complete for $r = 1$ but can be solved in polynomial time for $r \geq 2$.
The following questions remain unsolved:
\begin{enumerate}[{\bf {Question} 1:}]
	\item What is the complexity of \textsc{D$r$DSR} ($r \geq 2$) under $\sfTS$ on trees?
	\item What is the complexity of \textsc{D$r$DSR} ($r \geq 2$) under $\sfTS$ on interval graphs?
\end{enumerate}

\section*{Acknowledgments}
\ExtraMod{We thank the anonymous reviewers for their valuable comments and suggestions that helped improve eealier versions of this paper.
This work was partially completed while Niranka Banerjee was a member at the Research Institute of Mathematical Sciences, Kyoto University and Duc A. Hoang was working as a postdoctoral researcher at the Vietnam Institute for Advanced Study in Mathematics (VIASM).
Niranka Banerjee and Duc A. Hoang gratefully acknowledge the support and hospitality of RIMS and VIASM.}

\section*{Funding}
Niranka Banerjee was funded by Japan Society for the Promotion of Science (JSPS) KAKENHI Grant Number JP20H05967. 
Duc A. Hoang was partially funded by University of Science, Vietnam National University, Hanoi under project number TN.23.04 \ExtraMod{and by the Vietnam National University, Hanoi under the project QG.25.07 ``A study on reconfiguration problems from algorithmic and graph-theoretic perspectives''}.

\printbibliography

\end{document}